\documentclass{vldb}

\usepackage{balance}  % for  \balance command ON LAST PAGE  (only there!)

\usepackage{algorithm,algorithmicx,algpseudocode}
\usepackage{amsthm,caption}
\usepackage{amssymb,amsmath,amsthm,amsfonts}
\usepackage[usenames, dvipsnames]{color}
\usepackage[inline]{enumitem}
\usepackage{graphicx}
\usepackage{subcaption}
\usepackage{tikz}
\usepackage{multirow}
\usepackage{url}

\usepackage{tabularx}
\usepackage{dcolumn}
\usepackage{multirow}

\usepackage{soul}

\definecolor{lightblue}{rgb}{.75,.95,.95}
\sethlcolor{lightblue}
\sethlcolor{white}

\usepackage{multirow}
\usepackage{balance}
\usepackage{xspace}
\usepackage{hyperref}
\usepackage[capitalize]{cleveref}
\usepackage{outlines}
\usepackage{hyperref}

\usepackage{array}       
\usepackage{tabularx}
\usepackage{dcolumn} 
  \newcolumntype{d}[1]{D{.}{.}{#1|}}    
\usepackage{booktabs}

\newcommand{\eat}[1]{}

\newcolumntype{d}[1]{D{.}{.}{#1|}}

\newcommand{\dom}{\mathcal{T}}

\newcommand{\pol}{P}
\newcommand{\polMinRelaxSuffix}{mr}
\newcommand{\polMinRelax}{\pol_{\polMinRelaxSuffix}}
\newcommand{\polSensV}{0}
\newcommand{\polNSensV}{1}
\newcommand{\polExt}{\pol:\dom\rightarrow \{\polSensV , \polNSensV \}}
\newcommand{\polRelax}{\succeq_p}

\newcommand{\ospN}[2][]{N_{\pol_{#1}}(#2)}
\newcommand{\eospN}[2][]{N^e_{\pol_{#1}}(#2)}

\newcommand{\ospNMinRelax}[1]{\ospN[\polMinRelaxSuffix]{#1}}

\newcommand{\algoname}[1]{\textnormal{\textsc{#1}}}
\newcommand{\allsens}{all}

\newcommand{\PrTheta}{\Pr\mbox{$_\theta$}}

\newcommand{\mech}{\mathcal{M}}
\newcommand{\out}{\mathcal{O}}
\newcommand{\rr}[1][]{%
\ifthenelse{\equal{#1}{}}{ \algoname{OsdpRR}}{ OsdpRR_{#1}}%
}

\newcommand{\negLap}{Lap_-}

\newcommand{\domD}{\mathcal{D}}

\renewcommand{\vec}[1]{\mathbf{#1}}
\newcommand{\nvec}[1]{\widetilde{\mathbf{#1}}}
\newcommand{\minuseq}{\mathrel{-}=}

\newcommand{\tippers}{\mbox{TIPPERS}}

\newcommand{\funcDef}{f: \domD \rightarrow \mathbb{R}^k}

\theoremstyle{plain}
\newtheorem*{problemStatement*}{Problem Statement}

\newtheorem{definition}{Definition}[section]
\newtheorem{theorem}{Theorem}[section]
\newtheorem*{theorem*}{Theorem}

\newtheorem{lemma}{Lemma}[section]
\theoremstyle{definition}
\newtheorem{example}{Example}

\usepackage{array}
\newcolumntype{P}[1]{>{\centering\arraybackslash}p{#1}}

\begin{document}

% ****************** TITLE ****************************************

\title{One-sided Differential Privacy}

% ****************** AUTHORS **************************************

\numberofauthors{5} 

\author{
\alignauthor
Stelios Doudalis\\
\affaddr{U. C. Irvine}\\
\email{sdoudali@uci.edu}
\alignauthor
Ios Kotsoginannis\\
\affaddr{Duke University}\\
\email{iosk@cs.duke.edu}
\alignauthor
Samuel Haney\\
\affaddr{Duke University}\\
\email{shaney@cs.duke.edu}
\and
\alignauthor
Ashwin Machanavajjhala\\
\affaddr{Duke University}\\
\email{ashwin@cs.duke.edu}
\alignauthor
Sharad Mehrotra\\
\affaddr{U. C. Irvine}\\
\email{sharad@ics.uci.edu}
}

\maketitle

\begin{abstract}
In this paper, we study the problem of privacy-preserving data sharing, wherein only a subset of the records in a database are sensitive, possibly based on predefined \emph{privacy policies}. 
Existing solutions, viz, differential privacy (DP),  are over-pessimistic and treat all information as sensitive.
Alternatively, techniques, like access control and personalized differential privacy, reveal all non-sensitive records truthfully, and they indirectly leak information about sensitive records through \emph{exclusion attacks}.

Motivated by the limitations of prior work, we introduce the notion of \emph{one-sided differential privacy} (OSDP).
We formalize the exclusion attack and we show how OSDP protects against it.
OSDP offers differential privacy like guarantees, but only to the sensitive records.
OSDP allows the truthful release of a subset of the non-sensitive records.
The sample can be used to support applications that must output true data, and is well suited for publishing complex types of data, e.g. trajectories.
Though some non-sensitive records are suppressed to avoid exclusion attacks, our experiments show that the suppression results in a small loss in utility in most cases. 
Additionally, we present a recipe for turning DP mechanisms for answering counting queries into OSDP techniques for the same task. 
Our OSDP algorithms leverage the presence of non-sensitive records and are able to offer up to a $25\times$ improvement in accuracy over state-of-the-art DP-solutions.

\end{abstract}

%!TEX root = one_sided_privacy.tex
\section{Introduction}
\label{sec:intro}

%%%%%%% 
%%%%%%% General Introduction
%%%%%%% 

Public and private organizations capture and store diverse types of information about individuals that  are either  obligated to publish or  could benefit from sharing with others \cite{medicalSharingImperative}. 
A key impediment in sharing is the concern of privacy.
For instance, information about visited locations, online purchases, movie preferences \cite{netflixLeak}, or web searches \cite{aolLeak} has been shown to lead to inferences about sensitive properties like health, religious affiliation, or sexual orientation.
% , or family situation.

%%%%%%% 
%%%%%%% DP is the de-facto solution. Give intuition of DP. 
%%%%%%% 

In recent years, the field of privacy preserving data-sharing has flourished and \emph{differential privacy} \cite{fnt:DworkRoth14} (DP) has emerged as the predominant privacy standard.
Informally, DP states that for a database that contains a single record per user, the participation of an individual will change the output of an analysis by a bounded factor proportional to a privacy parameter $\epsilon$.
%In a database that contains a single record per user, DP states that the output of a differential private mechanism should change only by bounded factor proportional to a privacy parameter $\epsilon$, when a single record is added, removed or modified. 
Most work in differential privacy (with the exceptions of \cite{alaggan:jpc17, jorgensen:icde15}) has  been developed  under the assumption that every record in the database is \textit{equally} sensitive.

%%%%%%% 
%%%%%%%  We study a different problem 
%%%%%%%  Policies appear naturally in real life scenarios. 
%%%%%%%  Show with examples that can be cited that policies are real. 
%%%%%%% 

In this work, we depart from the above well trodden path, as we consider data sharing when only part of the data is sensitive, while the remainder (that may consist the majority) is non-sensitive. 
In real-world scenarios, the sensitive/non-sensitive dichotomy of data is not always apparent, but in certain cases it naturally occurs as a consequence of legislature, personal preferences, convenience, and privacy \emph{policies}.

%%%% First example legislature %%%%f

\begin{example}
The General Data Protection Regulation \cite{eu:gdpr, gdprConsent} (GDPR), that will take effect in European Union (EU) in 2018, imposes strict rules on how information about EU citizens is handled by companies.
Individuals must provide an affirmative consent (opt-in) in order to allow the collection and analysis of their data.
Additionally, information about minors under 16 years of age can not be stored and processed without parental authorization.
Failure to comply with GDPR could ensue a fine of up to 4\% of the company's annual turnover.
The previous legislative rules can easily be used to classify a citizen's record as sensitive or not. 
\end{example}

%%%% Second example personal preferences / convenience %%%%
\begin{example}
In psychology literature, it is well established that privacy attitudes depend on personal preferences and convenience \cite{ackerman:ec99,acquisti:sp05}.
In a large scale study in e-commerce \cite{berendt:cacm05}, researchers found two major groups of individuals: those who had little to no privacy concerns and the ``\textit{privacy fundamentalists}''. 
Moreover, individuals' tendency to opt-in or opt-out of data sharing is  a function of perception of utility offered by an application versus the cost of privacy leakage.
Thus, if the gain is sufficient, a big portion of users will willingly release their information.
\end{example}

%%%% Third example policies , smart building %%%%
\begin{example}
\label{exmpl:smart_builds}
%One motivating use case we explore is the application of privacy in the world of Internet of Things (IoT) and specifically in case of smart building management systems\cite{intelBIot} (SBMS).
We examine the use case of privacy applications in a smart building management system \cite{intelBIot}  (SBMS). Such systems, are capable of tracking individuals using indoor localization sensors and devices.
The collected information is  used to improve the performance of HVAC (heating, ventilation, and air conditioning) subsystems or offer location based services to the tenants, such as allowing to locate each other. 
SBMS need to comply with legislature, e.g. GDPR, and privacy policies. 
For instance, individuals should not be tracked in certain locations (e.g. restrooms, smoker lounge, etc.).
Additionally,  ``\textit{privacy fundamentalists}'' can opt-out from the tracking system, at the expense of enjoying only limited services. 
\end{example}

%%%%%%% 
%%%%%%% Access Control Could solve the problem, but exclusion attacks 
%%%%%%% 

%Query answering in the presence of policies has been extensively studied in the access control literature (e.g. tuple-level access control \cite{browder2002virtual, sandhu1996role}, positive and negative authorization views \cite{rizvi:sigmod2004}).

In this paper, we use policies as the language for specifying sensitivity.
Traditionally, query answering in the presence of policies has been the focus of access control literature \cite{browder2002virtual, rizvi:sigmod2004, sandhu1996role}.
The \emph{Truman} and \emph{non-Truman} models  \cite{rizvi:sigmod2004} are the two primary models for query answering in access control.
Based on the Truman model, queries are transparently re-written to ensure users only see what is contained in the authorized view.
Thus, the result is  correct only with respect to the restricted view.
In  non-Truman models queries that can not be answered correctly are rejected.
However, neither approach guarantees privacy. 
Consider example \ref{exmpl:smart_builds}, where the smoker's lounge is the only sensitive location and Alice is attempting to locate Bob.
If Bob is at the sensitive location, the Truman model will return nothing and the non-Truman model will honestly reject the query. 
Therefore, in both cases, Bob's location is indirectly revealed, because the sensitivity of his record is correlated with the value of the record. 
\hl{This attack is an example of what we call an \textit{exclusion attack}, which we formally define in section} \ref{sec:osp_exclusion_attack}.

  % \footnote{\small \hl{A  concept  similar to exclusion attack has also been observed  in a very different context of simulatable auditing} \cite{Kenthapadi:2005}.}

%\hl{This attack is an example of what we call an \textit{exclusion attack} and as a concept it has already been observed in the literature of simulatable auditing \cite{Kenthapadi:2005}.

%%%%%%% 
%%%%%%%  DP is an overkill
%%%%%%%  PDP , HDP could solve the problem, but exclusion attacks.
%%%%%%% 
 
An alternative solution is to directly apply differential privacy.
Such an approach would hide all properties of a record (including its sensitivity) and protect against exclusion attacks.
However, DP algorithms do not leverage on the fact that part of the data is non-sensitive to lower the error rates of  the query answers. 
In prior work, heterogeneous DP \cite{alaggan:jpc17} and personalized differential privacy (PDP) \cite{jorgensen:icde15} have tried to take advantage of different privacy levels between records.
Specifically, PDP \cite{jorgensen:icde15} considers the case where 
each record $r$ is associated with its own privacy parameter $\epsilon(r)$. 
PDP, for instance, allows suppressing records with strong privacy requirements $(\epsilon(r) < t)$, and analyzing only the records with weaker privacy requirements using a $t$-DP algorithm. 
In this context, if we model non-sensitive records by setting their privacy level to $\infty$, PDP would release all the non-sensitive records unperturbed (which is an $\infty$-DP algorithm), which results in an exclusion attack. 

%%%%%%% 
%%%%%%%  Overview why exclusion attacks happen
%%%%%%% 

In general, the fundamental problem  ignored by prior work is that an individual's privacy attitude (and hence the fact that their records is sensitive or non-sensitive) can be correlated with attributes in their record (and whether they are present or absent in the database).
For instance in \cite{milne:optin00}, the authors note that an individual's privacy preferences are correlated to demographic attributes like income, age, and political views.

%%%%%%% 
%%%%%%%  Contributions
%%%%%%% 

\noindent 
In this paper, we make the following contributions:

\begin{enumerate}[leftmargin=*,align=left]
	\item Motivated by the shortcomings of existing solutions, we introduce \emph{one-sided differential privacy} (OSDP), a novel privacy definition that provides rigorous privacy guarantees on the sensitive portion of the data, and guards against exclusion attacks.

	\item We develop a variant of the randomized response algorithm, that can truthfully release a sample of the non-sensitive data, while satisfying OSDP.
	We use this to release and analyze mobility trajectories of users in the context of smart buildings. 
	DP algorithms are unable to analyze these data with low error. In contrast, data released under OSDP supports these analyses with reasonably high accuracy, especially when the fraction of sensitive data is small. 	

	\item 
	We propose a new general recipe for adapting DP algorithms for answering counting queries to OSDP. 
	Our approach leverages on the non-sensitive records, and it is able to improve state-of-the-art DP algorithms by up to $25\times$ for certain inputs.
\end{enumerate}

%%%%%%% 
%%%%%%%  Paper organization
%%%%%%% 
%
The remainder of the paper is organized as follows. 
In Section \ref{sec:preliminaries}, we present differential privacy and its properties. 
We define one-sided differential privacy in Section \ref{sec:osp}, prove its composition properties, and discuss its relation to similar privacy definitions.
In Section \ref{sec:osp}, we also formalize the exclusion attack and we show how OSDP protects against it.
Section \ref{sec:release_true_data} contains an algorithm for releasing true data under OSDP.
In Section \ref{sec:release_aggregates}, we shift our focus on releasing counting queries under OSDP. 
In Section \ref{sec:experiments}, we empirically evaluate our algorithms. 
Section \ref{sec:extentions_future} contains interesting new directions of research and extensions to OSDP.
In Section \ref{sec:conclusion}, we summarize our contributions and experimental results.

%!TEX root = one_sided_privacy.tex
\section{Preliminaries}
\label{sec:preliminaries}
In this section, we introduce differential privacy (DP), its composition properties, and a basic mechanism that satisfies differential privacy. 

We adopt the bounded model of differential privacy \cite{tcc:DworkMNS06}. \hl{Let $\dom$ be the universe of records, $D$ be a database defined as a multiset of records $r \in \mathcal{T}$, and $\domD$ be the universe of all possible databases.}

%\begin{definition}[Bounded Neighbors]
\begin{definition}[Neighboring Databases]
\label{def:dp_bounded_neighbors}
 Databases $D$ and $D'$ are neighboring if they differ in the value of at most one record.
 That is $D' = D \backslash \{r\} \cup \{r'\}$.
 We denote the neighbors of $D$ as $N(D)$.
\end{definition}

\begin{definition}[Differential privacy \cite{tcc:DworkMNS06}]
  \label{def:differential-privacy}
 A randomized algorithm $\mech$ satisfies $\epsilon$-differential privacy if for all $D$ and $D'\in N(D)$ and all $\out \subseteq range(\mech)$:
 \begin{equation*}
    Pr[\mech(D) \in \out] \le e^\epsilon Pr[\mech(D') \in \out].
 \end{equation*}
 Where the probabilities are taken over coin tosses  of $\mech$.
\end{definition}

In DP, $\epsilon$ plays the role of a privacy knob -- with lower $\epsilon$ values  corresponding  to higher levels of privacy.
% DP mechanisms satisfy two important properties - sequential and parallel composition \cite{fnt:DworkRoth14}.

\begin{theorem}[Sequential Composition \cite{fnt:DworkRoth14}] 
\label{thm:seq}
  Let $\mech_1, \cdots \mech_k$ be a set of mechanisms such that $\mech_i$ satisfies $\epsilon_i$-differential privacy.
  Then $\mech_1  \circ \cdots \circ \mech_k$ satisfies $\sum_{i=1}^k \epsilon_i$-differential privacy.
\end{theorem}
A consequence of Theorem \ref{thm:seq} is that complex DP algorithms can be designed by sequentially composing simpler algorithms on the data.
For this reason, it is useful to consider the privacy parameter $\epsilon$ as a \textit{privacy budget} that we slowly spend by performing private analyses on the data.

%%%% Parallel composition DP

%
%\begin{theorem}[Parallel Composition]
%\label{thm:par}
%Let $D$ be a database and $\mech_1,\ldots,\mech_k$ be $\epsilon_1,\ldots,\epsilon_k$ differentially privacy algorithms respectively.  For $D_1, \dotsc, D_k$ $k$ disjoint subsets of $D$ i.e., $\forall i,j \in [k]: D_i \cap D_j = \emptyset$ and $\bigcup_{i\in [k]} D_i = D$, the sequence $\mech_i(D_i)$ satisfies $2\cdot \max_i \epsilon_i$-differential privacy
%\end{theorem}
%
%Theorem \ref{thm:par} states that applying differentially private algorithms on disjoint parts of the database comes with a small additional privacy cost.
%

\begin{definition}[Laplace Distribution]
\label{def:lap-dist}
The probability density function of the Laplace distribution with mean $\mu$ and scale $\beta$ is:
\begin{equation*}
f(x; \mu, \beta) = \frac{1}{2\beta}\exp\bigg(-\frac{|x - \mu|}{\beta} \bigg)
\end{equation*}
\end{definition}
For the remainder, we use $\mbox{Lap} (\beta) $ to denote a random variable drawn from the Laplace distribution with mean $\mu = 0 $ and scale $\beta$.

\begin{definition}[Sensitivity] Given a function $f: \domD \rightarrow \mathbb{R}$, the sensitivity of $f$ is:
$$ S(f) = \max\limits_{D, D'\in N(D)}\big\|f(D)-f(D')\big\|_{1}$$
\end{definition}
Sensitivity is a property of a numerical function $f$, as it is an upper bound to the output of $f$ between any two neighboring databases.

\begin{definition}[Laplace Mechanism \cite{tcc:DworkMNS06}] Consider a function $f: \domD \rightarrow \mathbb{R}^d$, for $d \in \mathbb{N}^+$, and database $D$. The  Laplace mechanism $\algoname{Laplace}$ is defined as $\algoname{Laplace}(D) = f(D) + \mathbf{z}$, where $\mathbf{z}$ is a vector of $d$ random variables s.t. $\forall i: z_i \sim \mbox{Lap}(S(f)/\epsilon)$.
\end{definition}

The Laplace mechanism satisfies $\epsilon$-differential privacy \cite{tcc:DworkMNS06} and is the standard mechanism for providing DP on tasks $f$ whose output is either a scalar or a vector of real numbers.

%!TEX root = one_sided_privacy.tex
\section{One-Sided Privacy}
\label{sec:osp}

In this section, we formally define  one-sided differential privacy (OSDP).
In subsection \ref{sec:osp_exclusion_attack}, we define the exclusion attack, and we prove that OSDP protects against it. Furthermore, in subsection \ref{sec:osp_composition}, we examine the composition properties of OSDP and in \ref{sec:osp_connection_with_other} we compare and contrast OSDP with related work.

\subsection{Definition}
\label{sec:osp_definition}

%The heart of OSDP are privacy policies that specify which records are sensitive and which are not. 

\begin{definition}[Policy Function]
	\label{def:policy}
	%Let $r \in \dom$ be a record.
	A policy function $\polExt$ denotes whether a record $r \in \dom$ is sensitive ($\pol(r) = 0$) or non-sensitive ($\pol(r) = 1$).
\end{definition}
A \emph{policy function} (or simply policy) classifies database  records as either sensitive or non-sensitive, examples of policy functions include: 
\begin{itemize}
	\item $\lambda r. if (r.Age \le 17): 0; else: 1$  encodes the policy that records corresponding to minors are sensitive. 
	\item $\lambda r. if (r.Race = NativeAmerican \vee r.Opt\-in = False): 0; else: 1$ implies that all records who have either opted out, or are native Americans are sensitive.
\end{itemize}
%Note that all policy functions depend on the value of  records in the database. 

Policies considered in this work are \textit{non-trivial} i.e., they result in at least one sensitive and one non-sensitive record. If all records are sensitive the private analysis can be done with standard DP algorithms. Conversely, if all records are non-sensitive then any non-private algorithm can be used.
%Private analysis under trivial policies can be done under either standard DP algorithms, or non-private algorithms. 

Our goal is to design a privacy definition that provides a rigorous privacy guarantee for the sensitive records (i.e., $\{r: \pol(r) = 0\}_{\forall r \in D}$). Towards that goal, we define one-sided neighboring databases under policy $P$.
\begin{definition}[One-sided $\pol$-Neighbors]
  \label{def:osp-neighbors}
\hl{Let $D$, $D'$ be two databases and $\pol$ a policy function. $D'$ is a one-sided $P$-neighbor of $D$, i.e., $D' \in \ospN{D}$, if and only if $\exists r \in D$ s.t., $P(r) = 0$, $\exists r' \in D'$ s.t., $r' \neq r$, and $D' = D \backslash \{r\} \cup \{r'\}$.}
\end{definition}
%  Let $D$, $D'$ be two databases and $\pol$ a policy function with $r \in D$ a sensitive record, i.e. $\pol(r) = \polSensV$, and $r' \in D' \backslash \{r\}$. % any possible record other than $r$.
%  Database $D'$ is a one-sided $P$-neighbor of $D$, i.e. $D' \in \ospN{D}$, if and only if $D' = D \backslash \{r\} \cup \{r'\}$.

%
One-sided neighbors under policy $P$ are created by replacing a single sensitive record with any other possible record. 
This implies that the $N_{\pol}$ relation is asymmetric, i.e., $D' \in \ospN{D}$ does not imply $D \in \ospN{D'}$. 
For instance, if $D$ contains only non-sensitive records then $\ospN{D} = \emptyset$.
On the other hand, $D$ is the $P$-neighbor of every database $D'$ with the same size as $D$,  that differs from $D$ by one sensitive record. 

\begin{definition}[One-sided Differential Privacy]
  \label{def:one-sided-privacy}
  Let $\mech$ be a randomized algorithm.
  Algorithm $\mech$ satisfies $(\pol, \epsilon)$-one-sided differential privacy, if and only if $\forall \out \subseteq range(\mech)$, and $\forall D, D' 
  \in \ospN{D}$:
  \begin{equation}
    \label{eq:osp-indistinguishability}
    \Pr [\mech(D) \in \out] \le e^{\epsilon} \Pr [\mech(D') \in \out]
  \end{equation}
 Where the probabilities are taken over coin tosses of $\mech$.
\end{definition}

% A mechanism satisfies OSDP, if the likelihood of any output does not change by too much when a sensitive record is replaced with any other record.
Essentially, OSDP  provides an indistinguishability property similar to DP, but only for the sensitive records.
As we show  in section \ref{sec:osp_exclusion_attack}, OSDP mechanisms not only ``hide'' the value of sensitive records, but also the fact that they are sensitive or not. 
% As we will see in the following subsection, this key result allows OSDP to protect against exclusion attacks.
%Also note, that for $D' = D \backslash \{r\} \cup \{r'\}$ and $\pol(r) = \pol(r') = 0$, OSDP gives exactly the same privacy guarantee as DP. 
Lastly, OSDP does not constraint the information leakage of non-sensitive records as long as the privacy of the sensitive records is preserved. 
\subsection{OSDP and the Exclusion Attack}
\label{sec:osp_exclusion_attack}

 In this subsection, we (a) formalize the notion of freedom from exclusion attacks, (b) show that prior work fails to satisfy it,  and (c) show that all one-sided differentially private algorithms satisfy it.

\noindent  \textbf{Exclusion Attacks:}
\hl{Informally, by exclusion attack we refer to attacks where adversaries decrease their uncertainty about the value of a sensitive record when that record is excluded from a private release. In example} \ref{exmpl:smart_builds}\hl{, if smoker's lounge is the only sensitive location, then excluding Bob's  data leaks information about his presence in the smoker's lounge. We formalize the concept of exclusion attacks  by formalizing its converse -- i.e., when a  mechanism is free of exclusion attacks. }

% by virtue of knowledge about the absence of that record from a private release.
% OLd text
%This approach can lead to privacy leaks on sensitive records, consider the case where an adversary knows Bob is in the dataset, but doesn't know Bob's value, if Bob's record is not released the adversary immediately learns that Bob has a sensitive value and can further infer Bob's value based on the policy used. 

% As discussed earlier, a common approach to share data with records of varying sensitivity levels  is to suppress all sensitive records and release all the non-sensitive record unperturbed. This approach can lead to privacy leaks on sensitive records. Consider the introductory example on SBMS, wherein the smoker lounge is the only sensitive location.
% Alice knows that Bob is in the database, but she is not aware of Bob's location. If Bob's location is not released, then she immediately learns that he is at the smoker's lounge based on the policy. 

% Formalize prior knowledge of adversary:
\hl{First we specify the background information of an adversary. Let $\theta$ be a probability distribution that describes the prior knowledge of an adversary. In general $\theta$ can be any arbitrary function that maps a database to a real number in $[0,1]$, i.e., $\theta: \dom^n \rightarrow [0,1]$. Informally, a mechanism $\mech$ satisfies freedom from exclusion attacks if an attacker's certainty about whether a target record is sensitive increases by at most a bounded amount after accessing $\mech$'s output. We now formalize \emph{freedom from exclusion attack}.}

%The \emph{problem} with the algorithm that releases all the non-sensitive records is that it does in fact disclose information about  sensitive records. First, revealing all the non-sensitive records can reveal whether or not individual's record is sensitive. 

%For instance, an adversary may know that a target record is in the dataset (based on a subset of attribute values). If the adversary does not see any released non-sensitive records matching the target record, then the adversary is certain that the target record is sensitive. Second, learning a record is sensitive is bad, since the fact that a record is sensitive or not is correlated with the values of the record (and in combination with side information can further reveal other properties of the individual). 

\begin{definition}[$\phi$-Freedom from Exclusion Attacks]
\label{def:freedom}
A mechanism $\mech$ satisfies $\phi$-freedom from exclusion attacks for policy $\pol$ and parameter $\phi$ if:
\begin{align*}\label{eq:freedom}
\lefteqn{\forall x: \pol(x)=0, \forall y \in \dom, \forall \out \subseteq range(\mech)}\\
\lefteqn{\forall \theta: \PrTheta(r = x) > 0, \PrTheta(r = y) > 0}\\
&& 
\frac{\PrTheta(r = x | \mech(D) \in \out)}{\PrTheta(r = y | \mech(D) \in \out)} \leq e^\phi \frac{\PrTheta(r = x)}{\PrTheta(r = y)}
\end{align*}
Where $r$ is the target record, $x$ is a value that makes the record sensitive, $y$ is another value in the domain, and $\theta$ is the adversary's prior about the database.
\end{definition}

The definition states that the adversary should not reduce his uncertainty about whether a record is sensitive after seeing the output of the mechanism by more than a multiplicative factor of $e^\phi$. It follows, that a mechanism that reveals all non-sensitive records is not free from exclusion attacks, since there are datasets and outputs where $P_\theta(r = y | \mech(D) \in O)$ becomes $0$ for values $y$ that are non-sensitive according to the policy, resulting in an unbounded posterior odds-ratio. 

\noindent \textbf{OSDP and Freedom from Exclusion Attacks:} Any mechanism that satisfies $(P, \epsilon)$-OSDP also satisfies $\epsilon$-freedom from exclusion attacks as long as the attacker does not know any correlations between the target record and the rest of the database. 
\hl{Suppose $\Theta$ denotes the set of probability distributions over databases of the form $\theta(D) = \prod_{t \in D} \theta(t)$, where $\theta(t)$ is the adversary's prior for tuple $t$. then we can show the following: }

\begin{theorem}\label{thm:osp-freedom}
	A mechanism satisfying $(\pol, \epsilon)$-OSDP satisfies $\epsilon$-freedom from exclusion attacks for policy $\pol$ \hl{for all adversaries whose prior can be expressed as a product of priors over individual records (i.e., $\theta \in \Theta$).}
\end{theorem}
\begin{proof}
	Since the adversary's prior can be decomposed into a product of priors over individual records, we can show that:
	\begin{align*}
	%\label{eq:freedom}
	\frac{\PrTheta(r = x | \mech(D) \in \out)}{\PrTheta(r = y | \mech(D) \in \out)} 
	 \le  \frac{\PrTheta(r = x)}{\PrTheta(r = y)} \max_{D, D'} \frac{\Pr_\mech(\mech(D) \in \out)}{\Pr_\mech(\mech(D') \in \out)}
	\end{align*} 
where $D$ is a database where record $r$ takes the value $x$ and $D'$ is a database where record $r$ takes the value $y$. These are precisely neighboring databases under OSDP (when $x$ is a sensitive record, and $y$ is an arbitrary record). Thus,
\begin{align*}
	\frac{\PrTheta(r = x | \mech(D) \in \out)}{\PrTheta(r = y | \mech(D) \in \out)}
	\le \frac{\PrTheta(r = x)}{\PrTheta(r = y)} \cdot e^\epsilon
\end{align*}
\end{proof}
\noindent
The above proof also shows that all $DP$ mechanisms also satisfy $\epsilon$-freedom from exclusion attacks for any policy. 

The assumption that the adversary does not know any correlations between the target record and the rest of the database is critical. 
Prior work has shown a ``No Free Lunch Theorem'' \cite{sigmod:KiferM11} that states that any mechanism that ensures properties of records are protected against all possible adversaries can not necessarily provide any useful information about the data. 
To illustrate this, consider an $\epsilon$-differentially private mechanism that releases the number of non-sensitive records (using the Laplace mechanism). 
Suppose the adversary knows that either all the records are  sensitive, or all the records are non-sensitive.
Under such a strongly correlated prior, the adversary can tell with high probability whether a target record is sensitive (when the noisy count is close to 0) or not. 
Thus, useful mechanisms can only ensure freedom from exclusion attacks under some assumption on the attacker prior. 
We prove freedom from exclusion attack under the independence assumption (which is the standard assumption used when proving semantic privacy properties of differential privacy \cite{pods:KiferM12}).  
Considering freedom from exclusion attacks when records are correlated is an interesting direction for future work.

\subsection{Composition}
\label{sec:osp_composition}

Much like DP, one-sided differential privacy also composes with itself. First we define policy relaxations, which will help us to compose OSDP instantiations with different policies.

\begin{definition}[Policy Relaxation]
\label{def:pol_containment}
Let $\pol_1$, $\pol_2$ be two policy functions.
$\pol_1$ is a relaxation of $\pol_2$, denoted as $\pol_1 \polRelax \pol_2$, if and only if for every record $r$, $\pol_1(r) \ge \pol_2(r)$.
\end{definition}

If $\pol_1$ is a relaxation of $\pol_2$, then every record that is sensitive under $\pol_1$ is also sensitive under $\pol_2$. 
We also say that $\pol_1$ is \emph{weaker} than $\pol_2$ or 
$\pol_2$ is \emph{stricter} than $\pol_1$.
Relaxing the policy also results in a weaker privacy definition.
 
\begin{theorem}[Privacy Relaxation]
\label{def:relaxation}
Let $\mech$ be a $(\pol_2, \epsilon)$-OSDP mechanisms, . If $\pol_1 \polRelax \pol_2$, then $\mech_2$ also satisfies $(\pol_1, \epsilon)$-OSDP.
\end{theorem}
\begin{proof}
If $\pol_1 \polRelax \pol_2$, then for every database $D$, $\ospN[1]{D} \subseteq \ospN[2]{D}$.
Thus based on Definition \ref{def:osp-neighbors}, $\mech_2$ also satisfies $(\pol_1, \epsilon)$-OSDP.
\end{proof}

We next define the minimum relaxation of a pair of policies that are not relaxations of each other. 

\begin{definition}[Minimum Relaxation]
Let $\pol_1, \ldots, \pol_k$ be a set of policy functions.
We define the minimum relaxation as the policy function $\polMinRelax$ such that for every record $r$, $\polMinRelax(r) = \max(\pol_1(r), \ldots, \pol_k(r))$. 
\end{definition}

The minimum relaxation of two policies, $\pol_1$, $\pol_2$ classifies as sensitive the records that are sensitive under both policies. 
Records that are deemed non-sensitive by either $\pol_1$ or $\pol_2$ are considered non-sensitive by $\polMinRelax(r)$. 
$\polMinRelax$ is the strictest policy that is a relaxation of both $\pol_1$ and $\pol_2$. When $\pol_1 = \pol_2$, then $\polMinRelax = \pol_1 = \pol_2$.  
We are now ready to define the sequential composition theorem for OSDP. 

\begin{theorem}[Sequential Composition]
  \label{thm:OSP_sequential}
  Let $\mech_1, \ldots, \mech_k$ be $(\pol_1, \epsilon_1), \ldots, (\pol_k, \epsilon_k)$-OSDP algorithms and 
  $\polMinRelax$ the minimum relaxation of $\pol_1, \ldots, \pol_k$.
  Then the composition of mechanisms $\mech_C = \mech_1 \circ \ldots \circ \mech_k$ satisfies ($\polMinRelax$, $\sum_i \epsilon_i$)-OSDP.
\end{theorem}
\begin{proof}
According to Theorem \ref{def:relaxation}, every algorithm $\mech_i$ satisfies ($\polMinRelax$, $\epsilon_i$)-OSDP.
Let $D$ be a database and $D' \in \ospNMinRelax{D}$, then
  \begin{align*}
    \frac{\Pr [\mech(D) = (x,y)]}{\Pr [\mech(D') = (x,y)]} 
    & = \prod_i \frac{\Pr [\mech_1(D) = x]}{\Pr [\mech_1(D') = x]} 
    &\le e^{\sum_i \epsilon_i}
  \end{align*}
Thus $\mech$ satisfies ($\polMinRelax$, $\sum_i \epsilon_i$)-OSDP.
\end{proof}

\subsection{Connection with other Privacy Definitions}
\label{sec:osp_connection_with_other}
In this subsection, we discuss the connection of OSDP with other three similar privacy definitions: differential privacy, personalized differential privacy and blowfish privacy. 

%%%%%%%%%%%
%%%%%%%%%%% Differential Privacy
%%%%%%%%%%%

\noindent {\bf Differential Privacy:} % \cite{fnt:DworkRoth14} - we've already added citation earlier.
OSDP is a generalization of differential privacy; the latter is a special case of the former when every record is sensitive as shown below. 

\begin{definition}[All Sensitive Policy]
	\label{def:dp-policy}
	Let $\pol_{\allsens}$ denote the policy function that classifies every record as sensitive.That is, $\pol_{\allsens}(r) = 0$, $\forall r$. We call this the \emph{all sensitive policy}. 
\end{definition}

\begin{lemma}
\label{thm:DP-compatibility}
\hl{Any $\epsilon$-differential private mechanism also satisfies $(\pol, \epsilon)$-OSDP for any policy function $\pol$.}
\end{lemma}
\begin{proof}
An $\epsilon$-differentially private algorithm satisfies $(\pol_{\allsens}, \epsilon)$-OSDP. 
Every policy $\pol$ is a relaxation of $\pol_{\allsens}$. 
Thus, by Theorem~\ref{def:relaxation}, a mechanism that satisfies $\epsilon$-differential private also satisfies $(\pol, \epsilon)$-OSDP.
\end{proof}

\begin{lemma}
\label{thm:DP-compatibility2}
\hl{
Any $(P_{all}, \epsilon)$-OSDP private mechanism also satisfies $\epsilon$-differential privacy.}
\end{lemma}
\begin{proof}
\hl{Neighboring datasets under DP (Def.} \ref{def:dp_bounded_neighbors}\hl{) are also $P_{all}$-one-sided neighbors (Def.} \ref{def:osp-neighbors}\hl{). Then any mechanism that satisfies $\epsilon$-indistinguishability between any two $P_{all}$-one-sided neighboring datasets also satisfies $\epsilon$-indistinguishability between any two neighboring datasets.}
\end{proof}

Differential privacy is the strictest form of OSDP  and any differential private algorithm will  also satisfy any version of  one-sided differential privacy for the same $\epsilon$. 

%%%%%%%%%%%
%%%%%%%%%%% Personalized Differential Privacy
%%%%%%%%%%%
% The fact that under PDP the adversary is allowed to know $\Phi$, explains how PDP is vulnerable to exclusion attacks.}

% Moreover, PDP requires not only that $\Phi$ is public but also that ``\textit{a record's privacy parameter must not indicate anything about their sensitive values}''\cite{jorgensen:icde15}, i.e., the values of $\Phi$ for each record in the dataset.

% Summary of PDP
\noindent {\bf Personalized Differential Privacy (PDP) \cite{jorgensen:icde15}:}
In PDP different records in the dataset are allowed different levels of privacy. This is modeled by a publicly known function $\Phi$ that associates records $r$ with  a personalized privacy level $\Phi(r) = \epsilon_r$ -- \hl{i.e., every record must publicly ``declare" its privacy level \emph{independent of its value}. However, in the privacy model of OSDP the fact that a record is sensitive or not is secret. This the key distinction between OSDP and PDP.}
	
\hl{To explain this further, while it seems one can model $(\pol,\epsilon)$-OSDP using PDP with function $\Phi_\pol$ that associates sensitive records with privacy parameter $\epsilon$, and non-sensitive records with a privacy parameter of $\infty$, we can construct mechanisms that satisfy $\Phi_\pol$-PDP, but not $(\pol,\epsilon)$-OSDP. For example, consider the following algorithm: $\algoname{Suppress}$ that picks a threshold $\tau$, suppresses all sensitive records with privacy parameter $\epsilon_r < \tau$, and then runs a $\tau$-differentially private computation on the rest of the records. Then for $\tau = \infty$ $\algoname{Suppress}$ will suppress all sensitive records and permit the release of the non-sensitive records with no perturbation. While $\algoname{Suppress}$ satisfies $\Phi_\pol$-PDP, since it is a special case of the $\algoname{Threshold}$ algorithm} \cite{jorgensen:icde15}, \hl{it fails to satisfy OSDP.}

\hl{A second key difference is that while all OSDP mechanisms enjoy $\epsilon$-freedom from exclusion attacks, PDP is a definition that allows algorithms that are susceptible to exclusion attacks. It naturally follows from Definition }\ref{def:freedom} \hl{that $\algoname{Suppress}$ with $\tau = \infty$ is vulnerable to exclusion attacks. Even in the case where one chooses a large but finite value for $\tau$, $\algoname{Suppress}$ satisfies Definition} \ref{def:freedom} \hl{only for $\phi = \tau$ and hence does not satisfy $(\pol,\epsilon)$-OSDP, which requires all mechanisms to satisfy freedom from exclusion attacks for $\phi = \epsilon$.}

\begin{theorem}
\label{thm:pdp}
\hl{The $\algoname{Suppress}$ algorithm with threshold  $\tau$ satisfies freedom from exclusion attack for $\phi = \tau$.}
\end{theorem}
\hl{Since $\algoname{Suppress}$ provides a $\tau$ indistinguishability property for the non-sensitive records, it follows that it only enjoys $\tau$-freedom of exclusion attacks. We defer details of the proof to the full paper.}

%%%%%%%%%%%
%%%%%%%%%%% Blowfish Privacy
%%%%%%%%%%%

\noindent {\bf Blowfish Privacy \cite{Haney15:vldb,He14:sigmod}:} 
Recent work has considered generalizations of differential privacy by modifying the notion of neighboring datasets. One example is Blowfish Privacy, where the notion of neighboring datasets is captured using a graph theoretic construct also called a \emph{policy}. Unlike an OSDP policy, which encodes whether a record is sensitive or not, a Blowfish policies specify which properties of records are sensitive (for instance, whether neighboring databases must differ in the entire record, or just in a single attribute of a record). While OSDP policies provide varied privacy protection to different records, Blowfish policies permit the same level of privacy for all records. Finally, unlike OSDP, in Blowfish privacy the resulting neighboring database relation is \emph{symmetric}. 
%The one-sided privacy neighbor relation $N_{\pol}$ (Definition~\ref{def:osp-neighbors}) is \emph{assymetric}. %For instance, for $N_{\pol}(D)$ is empty if $D$ contains only non-sensitive records, and is not empty if $D$ contains even one sensitive record. 

\noindent{\bf Other Related Work:}
In \cite{Böhler2017}\hl{ the authors consider a new privacy model for releasing exact value of outliers in a data-stream while preserving the privacy of the rest of the data. We conjecture that releasing exact values of outliers leads to vulnerability to exclusion attacks to the rest of the records and defer a full proof as future work.
Our notion of exclusion attacks is also reminiscent of how denial of query answering reveals information about records in the context of query auditing. To limit this leakage simulatable auditing has been proposed }\cite{Kenthapadi:2005}\hl{but the algorithms proposed in that work cannot be applied in our setting. Lastly, in }\cite{Cao:2010}, \cite{Cao:2012} \hl{the authors consider the private release of data where only a subset of the attributes are sensitive. This is in contrast to the data model we adopt where the sensitive/non-sensitive  distinction occurs on a per tuple basis.
Note that unlike OSDP, the privacy notions proposed in these papers are defined as properties of the output of a mechanism, do not support composition properties, and do not protect against adversaries with side information.}

\section{Releasing True Data}
\label{sec:release_true_data}
\label{sec:osp-rr}
%(and hence prevents exclusion attacks)
% %Examples of such tasks include the creation of \emph{extractive summaries} \cite{Li:2009, jie2014} that return a ``representative" sample of the input data. 
% Extractive summaries are common in applications such as image summarization and sharing, as well as, summarizing reviews in services such as Yelp and Amazon. 
% Other tasks, that can greatly benefit from a sample of true data, include complex analyses, e.g. classification, or releasing histograms over very large domains and with very large sensitivity. 
% In Section~\ref{sec:exp-results}, we will empirically show the advantages of using \algoname{OsdpRR} for the last two cases.

In this section, we introduce  \algoname{OsdpRR}, a randomized response based algorithm  that satisfies OSDP. \algoname{OsdpRR} works by releasing a true sample of the non-sensitive records in the database. 
The ability to output a true data sample, while ensuring privacy, enables a new class of privacy preserving analysis tasks that
require true data records.  Examples of such tasks include the creation of \emph{extractive summaries} \cite{Li:2009, jie2014}, classification, or releasing histograms over very large domains and with very large sensitivity. In Section~\ref{sec:exp-results}, we empirically show the advantages of using \algoname{OsdpRR} for the last two cases.

% Extractive summaries  return a ``representative" sample of the input data and are common in applications such as image summarization and sharing, as well as, summarizing reviews in services such as Yelp and Amazon.

%Other tasks (that require sample of true data) include text processing such as clustering, classification, sentiment analyses, etc. for text collections  —  tasks vital for large number of big data applications.
%In the era of big data, where the collected information is in the order of petabytes, 
%the ability to output a true sample of data while ensuring formal privacy guarantees 

%\ios{Make the examples here reflect the experiments we do in 6.1 and 6.2 - i.e., classification tasks that need true data + very high dimensional histograms where DP algorithms simply cannot run (due to representation problems). Then add forward pointer to the experiment section.}
%\ios{I don't think we need this sentence} Often, given the large size of data (e.g., order of petabytes), a sample of the data suffices to derive insights and understand trends.

\begin{algorithm}
\caption{\algoname{OsdpRR} ($D, \pol, \epsilon$)}
\label{algo:ospRR}
\begin{algorithmic}[1]
\State $S \leftarrow \{ \}$
\For {$r \in D$} \label{algo:ospRR:loop_start}
	\If {$\pol(r) = 1 $} \label{algo:ospRR:sub_start}
		\State $S \leftarrow S \cup \{ r \}$ with probability $1 - e^{-\epsilon}$
	\EndIf	\label{algo:ospRR:sub_end}
\EndFor \label{algo:ospRR:loop_end}
\State \Return $S$
\end{algorithmic}
\end{algorithm}

$\rr$ is described in Algorithm \ref{algo:ospRR}.
The inputs to $\rr$ are a database $D$, a policy function $\pol$, and the privacy parameter $\epsilon$. 
In lines \ref{algo:ospRR:loop_start} to \ref{algo:ospRR:loop_end}, $\rr$ iterates through every record in the database. 
%All sensitive records, i.e. $\pol(r) = 0$, are suppressed.  
Non-sensitive records, i.e. $\pol(r) = 1$, are added to the output set $S$ with probability equal to $1 - e^{-\epsilon}$. In all other cases the record is ignored. 

\begin{theorem}
	 \label{thm:ospRR-privacy} $\rr$ satisfies $(\pol, \epsilon)$-OSDP.
\end{theorem}
\begin{proof}	  
%	We show that equation~\eqref{eq:osp-indistinguishability} holds for all possible pairs of neighbors and outputs. 
	%As a first step, we prove Theorem \ref{thm:ospRR-privacy} when the size of $D$ is equal to 1. 	
	Let $D$ be a database with size equal to $1$, i.e., it contains a single record.	
	In that case we represent neighboring databases as records $r$, $r'$. 
	There are only two possible outputs, either release record $r$ or suppress it (that we denote as $\rr(r) = \varnothing$).	
	If $r$ is non-sensitive then $r$ has no neighbors (see Definition~\ref{def:osp-neighbors}).
	Thus, we are interested only in cases that $r$ is sensitive ($\pol(r) = \polSensV$) and $r'$ is any possible record, such that $r \ne r'$. 	
	Based on $\rr$'s possible outputs, we have the following cases: 

\noindent 	
\emph{Case 1: $\mathit{\rr(r) = r}$ and $\mathit{r}$ is sensitive}.
	If $r$ is sensitive $\Pr[\rr(r) = r] = 0$, thus equation \eqref{eq:osp-indistinguishability} holds trivially.

\noindent
\emph{Case 2.1: $\mathit{\rr(r) = \varnothing}$ and $\mathit{r}$, $\mathit{r'}$ are sensitive}.	
	\begin{equation*}
    	\frac{\Pr[\rr(r) = \varnothing]}{\Pr[\rr(r') = \varnothing]} = 1 \le e^{\epsilon}
	\end{equation*}
	
\noindent
\emph{ Case 2.2: $\mathit{\rr(r) = \varnothing}$ and $\mathit{r}$ is sensitive, $\mathit{r'}$ is non sensitive}.	
	\begin{equation*}
    	\frac{\Pr[\rr(r) = \varnothing]}{\Pr[\rr(r') = \varnothing]} = \frac{1}{e^{-\epsilon}} = e^{\epsilon}
	\end{equation*}

\noindent
This concludes the proof for the case that $|D| = 1$. 
Next, we extend the proof for any possible database size. 
	Let $D$, $D'$ be two neighbors, that differ on the $k^{th}$ record.
 	Let $S = \{s_1, s_2, \ldots \}$ be a possible output, with $s_i$ being the output for the $i^{th}$ record. 
 	Since $\rr$ makes independent decisions for each record, we have the same result whether $\rr$ is applied to the complete database or each record separately. 
 	Thus we have
 \begin{equation*}		
	\frac{\Pr[\rr(D) = S]}{\Pr[\rr(D') = S]} = 		
	\frac{\Pr[\rr(r_k) = s_k]}{\Pr[\rr(r_k') = s_k]}  
	\le e^{\epsilon}		
 \end{equation*}		
This concludes the proof.
\end{proof}

\begin{table}[t]
 \centering
 \caption{Percentage of released non-sensitive (ns) records using $\rr$ vs $\epsilon$.}
 \begin{tabular}{ P{3.5cm} | c | c | c }
  privacy parameter $\epsilon$ & $1.0$ & $0.5$ & $0.1$ \\
 \hline
  \% of released ns records & $\sim 63\%$ & $\sim 39\%$ & $\sim 9.5\%$\\
 \end{tabular}
 \label{tbl:osp-rr-perc}
\end{table}

%%%%%
%%%%% Sample size
%%%%%
The probability distribution of the sample size, released by $\rr$,  is the Bernoulli distribution with probability of success equal to $1 - e^{-\epsilon}$ and number of trials equal to the number of non-sensitive records. 
Table~\ref{tbl:osp-rr-perc} shows the expected percentage of released non-sensitive data as a function of privacy parameter $\epsilon$. 
For $\epsilon$ equal to $1.0$, the percentage is as high as $63\%$,
while for low epsilon values ($\epsilon = 0.1$), the percentage  drops to approximately $9.5\%$.

\section{Answering Counting Queries}
\label{sec:release_aggregates}

% Put overview of the subsection

% General introduction - what is sec 7 all about
We shift our focus on the task of releasing histograms over a dataset. We define a histogram query to be a set of counts defined over a non-overlapping partitioning of the dataset. One could think of a histogram query as: 
\begin{align*}
&\text{SELECT group, COUNT(*)}\\
&\text{FROM TABLE WHERE $\langle$condition$\rangle$}\\
&\text{GROUP BY $\langle$keys$\rangle$}
\end{align*}
where the output reports all groups with both zero and non-zero counts. 

Histogram estimation under differential privacy is a well researched area with rich literature to draw from. Hay et al \cite{DPBench} present a near-complete benchmark comparison of differentially private algorithms.  The algorithms examined in \cite{DPBench} are sophisticated, over-engineered algorithms that permit the release of accurate low-dimensional histograms -- specifically, 1- or 2-dimensional histograms with the number of bins in the thousands. As shown in Lemma \ref{thm:DP-compatibility}, any DP algorithm for histogram release will satisfy OSDP. In this section, we develop OSDP algorithms that  leverage on the presence of non-sensitive records to provide lower error than even the state-of-the-art DP algorithms for this task.

% This is a task that DP algorithms perform well on (especially for releasing lower dimensional histograms) and one could use any of the state-of-the-art DP algorithms to release counting queries under OSDP. In this section we explore whether non-sensitive records in a dataset allow OSDP algorithms to release counts with significantly lower error than DP algorithms. For the remainder we focus on the specific task of releasing histograms under OSDP. 

One would expect that in the presence of non-sensitive records histogram release becomes automatically easier, but this is not the case as we identify two main issues with releasing histogram estimates under OSDP. First, even in the presence of a \textit{single} sensitive record the global sensitivity of releasing a count is still $1$ (and the sensitivity of a histogram is still $2$). In Section~\ref{sec:primitives} we show that the sensitivity can be reduced by answering queries only on non-sensitive records. This helps when some of the bins in the histogram have exclusively sensitive records and some bins have exclusively non-sensitive records. We can then use a DP algorithm for the bins with sensitive records, and an OSDP algorithm (from Section~\ref{sec:primitives}) for bins with non-sensitive records. 
However, in most cases, every bin the histogram could have a mix of sensitive and non-sensitive records. This occurs, for instance, when the policy is derived from an opt-in/opt-out scenario. Another example is when the policy depends on a single attribute (like race), while the histogram is constructed on a different attribute (like age). This might result in big discrepancies between the distribution of non-sensitive records and the distribution of sensitive records and using just non-sensitive records as a proxy for the entire data will result in high error. For such cases, in Section~\ref{sec:alg-dev}  we present a general recipe for improving the error of state-of-the-art DP algorithms by leveraging non-sensitive records, and relaxing their privacy guarantee to OSDP. 

%When the distribution of non-sensitive records differs from the distribution of sensitive records, then using just non-sensitive records as a proxy for the entire data will result in high error. 

\eat{The organization of the current section is as follows: in section \ref{sec:primitives} we introduce a perturbation OSDP algorithm for answering counting queries under OSDP, in sections \ref{sec:counting-rel} and \ref{sec:helps} we examine under which conditions using OSDP algorithms provides definite accuracy gains over DP algorithms, and in section \ref{sec:alg-dev} we show a general methodology for creating OSDP variants based on algorithms proposed in the rich literature of DP.
Lastly, we refer the reader to section 6 for a comprehensive evaluation of our proposed methods.
%, and lastly in section \ref{sec:agg-exp} we present our experimental results.
}

\subsection{OSDP Primitives}

\label{sec:primitives}
\algoname{OsdpRR} can be used to release histograms by running the query on the sample of non-sensitive records output by \algoname{OsdpRR}. However, this approach can be worse than the Laplace mechanism in many cases. 
\begin{theorem}
	Let $n = |D|$ denote the number of records in database $D$, and $d$ denote the number of bins in the histogram query. Then the expected $L_1$ error of computing the histogram on the output of an \algoname{OsdpRR} algorithm that satisfies $(P, \epsilon)$-OSDP is higher than the expected $L_1$  error of Laplace mechanism that satisfied $\epsilon$-DP, whenever: 
	\begin{equation}
		n\cdot\epsilon > 2d \cdot e^\epsilon
	\end{equation}
\end{theorem}
\begin{proof}
The expected $L_1$ error of Laplace mechanism is $2d/\epsilon$. In the case of OSDP, error is due to (a) the suppression of all sensitive records, and (b) suppression of $n \cdot e^{-\epsilon}$ non-sensitive records. Even if the number of sensitive records tends to zero, the $L_1$ error of \algoname{OsdpRR} is at least $n \cdot e^{-\epsilon}$. 
\end{proof}
For instance, a 2-d histogram query over a uniform grid that divides each dimension into $100$ disjoint bins would have $d = 10^4$. When $\epsilon = .1$, \algoname{OsdpRR} will have higher error than the Laplace mechanism if $n > 2.2 \times 10^5$.

For the remainder we use $\mathbf{x}$ to denote the histogram over all records in $D$, and $\mathbf{x}_{ns}$ over the non-sensitive records of $D$ alone. Note that under OSDP, a database $D' \in N_P(D)$ is a neighbor of $D$ in one of two ways: (a) a sensitive record in $D$ is replaced with another sensitive record in $D'$, or (b) a sensitive record in $D$ is replaced with a non-sensitive record in $D'$. Thus, if $\mathbf{x}_{ns}$ ($\mathbf{x'}_{ns}$) denotes the output of a histogram query over the \textit{non-sensitive records} of  $D$ ($D'$, resp.), then the neighboring histograms differ in at most one count, and all counts in $\mathbf{x'}_{ns}$ are no smaller than the counts in $\mathbf{x}_{ns}$. Thus, we can release histogram counts on the non-sensitive records by adding one-sided Laplace noise to ensure OSDP. The one-sided Laplace distribution, $\negLap(\lambda)$, is the mirrored version of exponential distribution with scale equal to $\lambda$, where all the mass is on the negative values.

\begin{definition}[One-Sided Laplace Distribution]
  \label{def:osp-noise}
  The probability density function of the One-Sided Laplace Distribution is:
    \begin{equation*}
    \label{def:osp-pdf}
    f_{\negLap}(x; \lambda) =
    \begin{cases}
      \frac{1}{\lambda}\exp\left(\frac{x}{\lambda}\right) & \text{ if $x \le 0$} \\
      0  & \text{otherwise}
    \end{cases}
  \end{equation*}
\end{definition}

%%% EATING SENSITIVITY
\eat{ We can define sensitivity under OSDP analogous to global sensitivity under DP. 

\begin{definition}[$P$-Sensitivity]
Let $D, D'$ be neighboring databases under policy $\pol$ i.e., $D' \in \ospN{D}$.
Let $\funcDef$ be a function, where $\domD$ is the universe of databases and $k \in \mathbb{N}^+$. Then the $\pol$-sensitivity of $f$ is:
\begin{equation*}
	S_P(f) = \max_{D, D' \in \ospN{D}} \big\| f(D) - f(D') \big\|_1
\end{equation*}
\end{definition}
} %%%%%

The \algoname{OsdpLaplace} answers queries by adding one-sided noise to the query answer computed only using non-sensitive records. 

\begin{definition}[\algoname{OsdpLaplace} ]
Let $D$ be a database, $\pol$ a policy, $D_{ns} =  \{ r \in D ~|~ \pol(r) = 1\}$ the subset of non-sensitive records. The $\algoname{OsdpLaplace}$ answers a histogram query with $d$ bins, by computing the histogram on the non-sensitive records $D_{ns}$, and adding a vector of  $d$ i.i.d. random variables drawn from $\negLap(\lambda)$. 
\end{definition}

\begin{theorem}\label{thm:osplap}
For $\lambda = 1/\epsilon$ \algoname{OsdpLaplace} satisfies $(\pol, \epsilon)$-OSDP.
\end{theorem}

\begin{proof}
	Let $D, D'$  $P$-neighboring databases. Let $\vec{x}$ ($\vec{x}'$) denote the true {\bf non-sensitive} histogram on $D$ ($D'$).
	For simplicity, let $\mech$ be an alias for \algoname{OsdpLaplace}.
	Note that  $range(\mech(D, \pol))$ $\subseteq range(\mech(D', \pol))$ since $\forall i$,  $x_i \le x'_i$. 
	Let $\vec{y} \in range(\mech(D', \pol))$, then there are two cases: 
	
	\noindent
	{\bf Case 1. $\exists i: x_i < y_i \le x'_i$}.
	Then, $\Pr [\mech(D) = \vec{y}] = 0$, thus equation \eqref{eq:osp-indistinguishability} holds trivially.
	
	\noindent	
	{\bf Case 2. $\forall i : y_i \le x_i$}
	\begin{align*}
		\label{}
		\frac{\Pr [\mech(D) = \vec{y}]}{\Pr [\mech(D') = \vec{y}]} 
		& = \frac{ 
		\prod\limits_{i} \exp\left( \epsilon\cdot(y_i - x_i)\right)}
		{\prod\limits_{i} \exp\left(\epsilon\cdot(y_i - x'_i)\right)} \\		
		& \le \exp\left( \epsilon \cdot  \| \vec{x}' - \vec{x} \|_1 \right)
		\le \exp(\epsilon)
	\end{align*}
\end{proof}

It is easy to see that the \algoname{OsdpLaplace} mechanism introduces noise with $1/8$ the variance of the differentially private Laplace mechanism. The exponential distribution has half the variance, and another factor of 4 reduction in variance comes from the sensitivity dropping from 2 (in the case of DP) to 1 (in the case of OSDP). The error can be further reduced by using the following algorithm that we call \algoname{OsdpLaplaceL1} that leverages the fact that all input counts are non-negative.

%\begin{algorithm}
%\caption{\algoname{OsdpLaplace} ( $\vec{x}_{ns}, \epsilon$)}
%\label{algo:osp-laplace}
%\begin{algorithmic}[1]
%\State $\nvec{x}_{ns} = \vec{x}_{ns} + \negLap (2 / \epsilon)^d$
%\State $\nvec{x}_{ns}[\nvec{x}_{ns} < 0.0] = 0.0$ \label{algo:osp-laplace:rm-zeros}
%\State Return $\nvec{x}_{ns}$
%\end{algorithmic}
%\end{algorithm}
%
%Line \ref{algo:osp-laplace:rm-zeros} of Algorithm \ref{algo:osp-laplace} is a standard post-processing step that sets all negative frequencies to $0$.

%We also propose a variant of $\algoname{OsdpLaplace}$   that uses two additional post-processing steps. The first ensures non negativity and the second takes into account the fact that we add only negative noise to the counts. From each positive-valued noisy count we subtract the median value of the noise added (i.e., $med(\negLap(\beta)) = - \ln(2) \beta$). We can easily show that this additional step decreases the $L_1$ error of the noisy counts.

\begin{algorithm}
\caption{\algoname{OsdpLaplaceL1} ($\vec{x}_{ns}, \epsilon$)}
\label{algo:osp_laplace_l1}
\begin{algorithmic}[1]
\State $\nvec{x}_{ns} = \vec{x}_{ns} + \negLap (1 / \epsilon)^d$
\State $\nvec{x}_{ns}[\nvec{x}_{ns} < 0.0] = 0.0$ 
\State $\mu = - \ln(2) \cdot 1/\epsilon $
\State $\nvec{x}_{ns}[\nvec{x}_{ns} > 0.0] \minuseq \mu$ 
\State \Return $\nvec{x}_{ns}$
\end{algorithmic}
\end{algorithm}
Step 2 of the algorithm sets all negative counts after adding noise to 0. Note that after this step, every count in $\mathbf{x}_{ns}$ that was originally 0 is also output as 0. However, there could be other bins that are output as 0 (since their noisy count was negative).  Next, note that the one-sided Laplace distribution is not unbiased, step 4 adds back the median value $\mu$ of a one-sided Laplace random variable to make the positive noisy counts unbiased. While \algoname{OsdpLaplace} and \algoname{OsdpLaplaceL1} add significantly lower noise than the Laplace mechanism, they can have higher overall error than the Laplace mechanism since they only use the non-sensitive records to answer queries.

\subsection{A Recipe for OSDP Algorithms}
\label{sec:alg-dev}
We now explore how to utilize the rich literature of DP algorithms in order to build new OSDP algorithms that take advantage of the presence of non-sensitive records in order to reduce the added noise when applicable.

We focus on an important class of DP algorithms for histogram release that can be abstracted to two distinct phases. In the first phase they query a set of statistics on the data and learn an underlying model of it.  Then they use the learnt model  and the $\algoname{Laplace}$ mechanism to add noise to a set of associated aggregate counts  which will be used to estimate the private histogram.  Examples of these algorithms include: \algoname{PrivBayes}\cite{Zhang2014} for high dimensional histograms, \algoname{AGrid}\cite{AGrid} for 2 dimensional histograms, \algoname{AHP}\cite{AHP} for general histograms, and \algoname{DAWA} for 1 and 2 dimensional histograms. Algorithms of this type can be extended under OSDP in the following way.

% A general recipe for extensions
The naive extension  replaces the noise addition step with an OSDP primitive and changes the input histogram to $\mathbf{x}_{ns}$. This extension is easy to implement as it treats the underlying DP algorithm as a black box, requiring from the user just a minor change. However, it suffers from the same problem other OSDP primitives do: if $\mathbf{x}$ and $\mathbf{x}_{ns}$ are dissimilar then the incurred error will be high.

%
%\hl{The second extension  combines the  techniques of established DP algorithms for releasing histograms with a zero count detection method from OSDP algorithms.}

\hl{We now describe our second, more elaborate extension that enhances the performance of established DP algorithms for releasing histograms.}  Let A be a two-phase algorithm as described earlier, we extend it as follows. 
First, we dedicate a fraction $\rho$ of the privacy budget to release the histogram of $\mathbf{x}_{ns}$ using an OSDP primitive from section \ref{sec:primitives}. We use this noisy histogram to compute the set of zero count bins $Z$. Then, we use the rest of the privacy budget to run the algorithm A. At the end of this process we get a set of zero bins $Z$ generated under $(P, \rho\cdot\epsilon)$-OSDP, and the noisy output of A under $(1-\rho)\epsilon$-DP. Before releasing the final noisy histogram, we add a post-processing step that (a) sets to zero all the counts of bins in $Z$, (b) uses the learned model to update the rest of the counts accordingly i.e., we redistribute the value of the counts we removed. By sequential  composition the entire algorithm satisfies $(P, \epsilon)$-OSDP. 

We apply this second technique to extend the DAWA algorithm, the state-of-the-art DP algorithm for this task \cite{DPBench}. We leave extensions of other algorithms as an interesting direction of future work. We call this new OSDP algorithm \algoname{DAWAz} (Alg. \ref{algo:dawaz}).
\algoname{DAWAz} executes an OSDP subroutine to get the zero count cells $Z$ and executes $\algoname{DAWA}$ on the full histogram to get $\bar{\mathbf{x}}$. Then,  $\algoname{DAWAz}$ examines the noisy estimate returned from the OSDP subroutine, finds the  bins with 0 count and replaces their values on $\bar{\mathbf{x}}$ with 0. Lastly, for every DAWA partition $\algoname{DAWAz}$ finds the number of replaced bins in that partition and reallocates the mass total count of the partition to the non replaced bins. 

\begin{theorem}
	\algoname{DAWAz} satisfies $(P, \epsilon)$-OSDP.
\end{theorem}
\begin{proof}
	Running the OSDP algorithm satisfies $(P, \rho\cdot\epsilon)$-OSDP, while the histogram construction using DAWA satisfies $(1-\rho)\epsilon$-DP. The rest of the algorithm is post-processing. Thus, the proof follows from sequential composition theorem. 
\end{proof}

\begin{algorithm}
\caption{$\algoname{DAWAz}(\mathbf{x}, \mathbf{x}_ns, \epsilon, \rho)$}
\label{algo:dawaz}
\begin{algorithmic}[1]
\State $\epsilon_1 \leftarrow \rho \cdot \epsilon$
\State $\epsilon_2 \leftarrow (1 - \rho)\cdot \epsilon$
\State $\tilde{\mathbf{x}}_{ns} \leftarrow \Call{Osdp}{\mathbf{x}_{ns}, \epsilon_1} $
\State $\bar{\mathbf{x}} \leftarrow \Call{DAWA}{\mathbf{x}, \epsilon_2}$
\State $Z = \{ i ~|~ \tilde{\mathbf{x}}_{ns}  = 0\}$
\State $\forall i \in Z: \bar{x}_i = 0$
\State Let $\mathcal{B}$ the set of partitions returned from \algoname{DAWA}.\
\For {$\forall B \in \mathcal{B}$}
	\State $rescale\_ratio = \frac{|B|}{|Z \cap B |}$
	\State $\forall i \in B: \bar{x}_i = \bar{x}_i \cdot rescale\_ratio$
\EndFor
\State \Return $\bar{\mathbf{x}}$
\end{algorithmic}
\end{algorithm}

\section{Experiments}\label{sec:experiments}

We use real and benchmark datasets to empirically demonstrate that the OSDP algorithms presented in this paper permit accurate data analysis of private datasets by leveraging  the presence of non-sensitive records. We show:
\begin{itemize}[leftmargin=*,align=left]
\item \algoname{OsdpRR} allows the release of true data records, which permits building classifiers with high accuracy \emph{close to that of a non-private baseline} (Section~\ref{sec:exp-results-classification}). In contrast, a state of the art DP classifier (that considers all the data as sensitive) has accuracy close to that of a baseline that does not have access to the data.
\item \algoname{OsdpRR}'s output also permits releasing high dimensional histograms with \emph{an order of magnitude smaller error} than that of a DP algorithm on the same data (Section~\ref{sec:exp-results-high})
\item We also consider answering low dimensional histogram queries, a task where over-engineered DP algorithms achieve highly competitive error rates. We show that  OSDP algorithms outperform state of the art DP algorithms  (Section~\ref{sec:exp-results-low}). In a comparison of 4 OSDP algorithms and 2 DP algorithms on a set of benchmark datasets, \algoname{DAWAz} has on average less than $2\times$ the error compared to the optimal error achieved by any one of these algorithms. In contrast, the state-of-the-art DP algorithm for this task \algoname{DAWA}, incurs $6\times$ the error of the optimal (Fig.~\ref{fig:MRE-RP-close}). In some cases we observe that OSDP algorithms achieve up to $25\times$ lower error rate than DAWA  (Fig.~\ref{fig:sparse99}).

%\item For answering low dimensional histogram queries, where over-engineered DP algorithms exist that can achieve low error, we show that OSDP algorithms outperform state of the art differentially private algorithms (Section~\ref{sec:exp-results-low-real}). In fact, on a set of benchmark datasets OSDP algorithm $\algoname{DAWAz}$ has at least \emph{a factor of $6\times$ lower mean relative error} than DP algorithm $\algoname{DAWA}$ for the task of releasing 1-D histograms (Section~\ref{sec:exp-results-low-benchmark}). 
\end{itemize}

We describe the datasets and the policy functions used in Section~\ref{sec:exp-data}, the analysis tasks in Section~\ref{sec:exp-tasks}, and the results in Section~\ref{sec:exp-results}.

\subsection{Datasets and Privacy Policies}\label{sec:exp-data}
%To compare DP and OSP in both settings described above, 

We use two datasets: $\tippers$, a real dataset from a live IoT testbed under development at UCI \cite{tippers}, and 
%\footnote{\url{http://tippersweb.ics.uci.edu/}
DPBench-1D, a set of $7$  one-dimensional real histograms, used to benchmark the performance of DP algorithms \cite{DPBench, kotsogiannis2017pythia}.

\subsubsection{$\tippers$} \label{sec:data-tip}
\noindent \textbf{Dataset:} The $\tippers$ dataset consists of  a trace of distinct devices
connected to 64  Wi-Fi access points located in the Bren Hall building at UC Irvine collected over a 9 month period from January to September of 2016. 
Each Wi-Fi trace consist of a triple $\langle$AP\_mac, 
device\_mac, timestamp$\rangle$ \footnote{the device\_mac is pseudo-anonymized to meet the IRB requirements for this study} and we use them to construct a total of $585$K unique daily trajectories with $16$K unique users. Users might connect to the same access point multiple times, or might cross areas covered by multiple access points. To address this, we discretized time to 10 minute intervals and set the location at each interval as the most frequent access point during that time. 

% One challenge arising from real sensor data is the high level of noise.
% Additionally, we split the original data into daily trajectories.

%One of the  challenges of working with real life sensor data is the high level of noise. As users move through the building, they might connect to the same access point multiple times, or they might cross areas that are covered by multiple access points.

%In Table \ref{tbl:wifi_general_stats}, we present some general statistics about our dataset.  

%\begin{table}[t]
%\centering
%\caption{Wi-Fi dataset statistics}
%\begin{tabular}{ l | c ||  l | c }
%Statistic &  & Statistic &\\
%\hline
% \# access points & 64 &
% \# users & $\sim$16K \\
% \# records & $\sim$585K  &
%\end{tabular}
%\label{tbl:wifi_general_stats}
%\end{table}

%To classify datasets into sensitive/non-sensitive, we devised based on access points. 
%The different levels of sensitivity, i.e. \% of sensitive records, are created by varying the number of sensitive access points. 
%At the first level of sensitivity, only a single access point is marked as sensitive. 
%In the second level, two access points are sensitive and so forth.
%The order that the access points are chosen is based on how ``busy'" an access point is, i.e. the number of daily trajectories that cross it. 
%Therefore, at the $k^{th}$ level of sensitivity, the $k$ least ``busy" access points are marked as sensitive.

\noindent  \textbf{Privacy Policy:}
% Thus, under DP, neighboring databases differ by arbitrarily changing one user's daily trajectory on a specific day. 
%and classifies all trajectories passing at least one time into a region having a sensitive access point as sensitive. 
 We consider a user's daily trajectory as the object of privacy protection.
Thus, neighboring databases under DP differ in a single user's daily trajectory on a specific day. This protects all properties of an individual within the time period of a specific day. We use access-point level policies to partition the dataset to its sensitive and non-sensitive parts. More specifically, our policies assume a sensitive set of access points (e.g., lounge or restroom) and classify as sensitive all trajectories that pass at least once through a sensitive access point. Based on this general recipe we define a policy $P_\rho$, that results in a non-sensitive dataset with $\rho/100$ share of non-sensitive records (e.g., $P_{99}$ results in a dataset with 99\% non-sensitive daily trajectories). \hl{In our experiments we use $\rho = \{99, 90, 75, 50, 25, 10, 1\}$ to capture different fractions of non-sensitive records.}

\subsubsection{DPBench-1D}
\label{sec:data-benchmark}

\begin{table}[t]
\caption{\label{table:Datasets} Histogram Benchmark}
\begin{center}
\small
	\begin{tabular}{ |l|c|r| }
	\hline
	Dataset & Sparsity & Scale\\  \hline \hline
	Adult & 0.98 & 17,665   \\ \hline
	Hepth  & 0.21 & 347,414   \\ \hline
	Income  & 0.45 & 20,787,122   \\ \hline
	Nettrace  & 0.97 & 25,714   \\ \hline
	Medcost  & 0.75 &  9,415  \\ \hline
	Patent & 0.06 &  27,948,226   \\ \hline
	Searchlogs & 0.51 & 335,889   \\ \hline
\end{tabular}
\end{center}
\end{table}
%
%We also evaluate our algorithms on a set of benchmark datasets that have been heavily utilized in prior work in DP \cite{DPBench} \cite{kotsogiannis2017pythia}. The datasets, shown in table~\ref{table:Datasets}, consist of $1$-d histograms with domain size (i.e., number of bins) equal to $4,096$ and they differ in properties like their sparsity, their scale (i.e., the total count of each histogram), and their partitionality. 
%%
%  In contrast with $\tippers$ where we used value based policies, we aim to use complex policies that result in an unclear partition of the data (e.g., similar to an opt-in/opt-out scenario). Towards that goal we created two novel biased sampling techniques: \algoname{MSampling} and \algoname{HiLoSampling}. These techniques are given a histogram $\mathbf{x}$ and return a non-sensitive histogram $\mathbf{x}_{ns}$ that follows either a Close (\algoname{MSampling}), or a Far (\algoname{HiLoSampling}) policy. Details of the sampling techniques are omitted due to space limitations.  For each our two policies, we used $4$ values of $\rho_\mathbf{x}$: $\{0.99, 0.90, 0.75, 0.50\}$. Thus, for each original histogram we have $4\times 2 = 8$ different non-sensitive histograms for a total of $56$ unique $(\mathbf{x}, \mathbf{x}_{ns})$ pairs to evaluate our algorithms on.
 
\noindent \textbf{Dataset:} We use the DPBench benchmark datasets \cite{DPBench} to evaluate the performance of our algorithms for the task of releasing 1-D histograms. This benchmark has a collection of 7 datasets. Tuples in each of these datasets are drawn from a categorical domain with  size 4,096. The datasets are of varying scales (i.e., number of records), sparsity (i.e., fraction of the total domain that does not appear in the active domain of the dataset) and shape (i.e., empirical distribution), as detailed in  table~\ref{table:Datasets}.

\noindent \textbf{Privacy Policy:} These datasets were originally used to benchmark DP algorithms, and hence do not have any notion of a policy. We simulate opt-in/opt-out policies by sampling a subset of the records as non-sensitive. We denote by $\mathbf{x}$ the true histogram (on $d = 4096$ bins), and by $\mathbf{x}_{ns}$ the histogram over the non-sensitive records. We use two biased sampling methods to generate $\mathbf{x}_{ns}$: \algoname{MSampling}, where the empirical distribution of $\mathbf{x}_{ns}$ is close to that of $\mathbf{x}$, and \algoname{HiLoSampling}, where the empirical distribution of $\mathbf{x}_{ns}$ significantly differs from that of $\mathbf{x}$. \algoname{MSampling} simulates the scenario when an individual's opt-in/opt-out preference has a low correlation with their value, while \algoname{HiLoSampling} simulates the case where the privacy preference is highly correlated with the individual's record value. We call the former policy {\em Close}, and the latter policy {\em Far}.

\algoname{MSampling} takes as input a histogram $\mathbf{x}$, a  sampling ratio $\rho_{\mathbf{x}}$, and a parameter $\theta$, and outputs $\mathbf{x'}$, a sample of $\mathbf{x}$ with $\|\mathbf{x}'\|_1 = \rho_\mathbf{x} \|\mathbf{x}\|_1$. The mean $\mu'$ and the standard deviation $\sigma'$ of $\mathbf{x'}$ are within a $1 \pm \theta$ multiplicative factor of the mean $\mu$ and standard deviation $\sigma$ of $\mathbf{x}$. We use $\theta = 0.1$. 

\algoname{HiLoSampling} works differently, its input is again the histogram $\mathbf{x}$ and the sampling ratio $\rho_{\mathbf{x}}$ and two additional parameters $\gamma > 1$ and $\beta$.  This technique partitions the domain of the histogram into ``High'' and ``Low'' regions. The algorithm picks a bin $b$ at random and sets the bins in the range $b \pm d \cdot\beta$ as ``High'' bins.  Every bin not in the High region is automatically in the Low region. Then, we proceed by sampling with replacement from bins of the histogram, bins of the ``High'' area are sampled with weight $\gamma$ , while bins of the ``Low'' area are sampled with weight $1$. Choosing a high $\gamma$ and a smaller $\beta$  allows us to create an $\mathbf{x}_{ns}$ that is highly dissimilar from $\mathbf{x}$. We set $\gamma = 5$ and $\beta = .4$.

 \hl{For each policy, we use $7$ different values of $\rho_\mathbf{x}$:} $\{0.99, 0.90,\allowbreak 0.75,$\hl{$ 0.50, 0.25, 0.10, 0.01\}$}\hl{. Thus, for each of the 7 datasets we have $7\times 2 = 14$ different non-sensitive histograms for a total of $98$ $(\mathbf{x}, \mathbf{x}_{ns})$ pairs to evaluate our algorithms on.}

%!TEX root = one_sided_privacy.tex

\subsection{Analysis Tasks}\label{sec:exp-tasks}

\vskip 3pt \noindent  {\bf Classification:} 
First we consider a classification task on the $\tippers$ dataset:  based on a user's  daily trajectory we classify whether the user is a \textit{resident} or a \textit{visitor}. More specifically, we learn a logistic regression (LR) model to predict whether a daily trajectory corresponds to a resident. Due to IRB restrictions, we do not have access to the true identities of the owners of each device\_mac in the data.  For that reason, to create the training data we classify as residents owners of device\_mac that (a) visit Bren Hall at least 10 days per month for the last 5 months and either (b) work beyond 7pm once a week, or (c) work more than 6 hours at once per week; every other user is classified as a visitor. In total, we detected $381$ residents that correspond to $\sim 43K$ daily trajectories (out of $\sim 553K$ total daily trajectories).
 We use the following features derived from the daily trajectories:  duration of stay at Bren Hall, distinct number of access points visited during a day, and the number of times each individual access point is visited in a trajectory. We also consider frequent patterns of the form $(AP_1, AP_2, AP_3)$ from the trajectories. A trajectory satisfies a frequent pattern if it visits the three access points during the day at consecutive time intervals. We chose patterns that appear in at least 50 trajectories and construct a feature for each pattern using the number of times the pattern appears in a daily trajectory. 

We evaluate the performance of all LR classifiers using the receiver operating characteristic (ROC) curve, that is created by plotting the true positive rate versus the false positive rate for multiple thresholds over the probability of predicting a test record as belonging to a resident. The accuracy is usually reported using the average area under the ROC curve (AUC) over 10-fold cross validation, which is a value between $[0,1]$.  We report $1.0 - \algoname{AUC}$ as a measure of prediction error.

\vskip 3pt \noindent  {\bf High Dimensional Histograms:}
We also consider the task of privately releasing high dimensional histograms on the \tippers dataset.
More specifically, we consider the number of distinct users for every n consecutive access points in a trajectory (i.e., the n-gram). This is a highly non-trivial task since (a) the domain size for an n-gram is $64^n$ and (b) the histogram over all bins has sensitivity $64^n$ since a user might appear in all n-grams. 
We use the mean relative error (MRE) to measure the quality of our private n-gram estimates. 
More specifically, for a histogram $\mathbf{x}$ of size $d$ and its private estimate $\mathbf{\tilde{x}}$, we have: 
\begin{align*}
		\algoname{MRE} (\mathbf{x}, \mathbf{\tilde{x}})  = \frac{1}{d}\displaystyle\sum_{i=1}^d  |x_i - \tilde{x}_i|/\max (x_i, \delta)
\end{align*}
In our experiments, we set $\delta = 1$ and report the average MRE over $10$ independent executions.

%!TEX root = one_sided_privacy.tex
\vskip 3pt \noindent {\bf Low Dimensional Histograms:}
We also evaluate our algorithms on releasing  $1$ and $2$-D histograms under OSDP. Our goal is twofold, we want to: evaluate the performance of OSDP algorithms in a realistic scenario and  explore under which conditions OSDP algorithms offer advantages over their DP counterparts. Towards these goals we use histograms from both the benchmark datasets and \tippers. More specifically,  we used  \tippers to generate a 2-D histogram that contains the number of distinct users that connected to every access point and for every hour of a single day. For the DPBench-1D datasets, we consider the task of releasing 1-D histograms. 
% Error Metrics
In addition to MRE defined earlier, we also use per bin-relative error to measure the utility of the algorithms. Per-bin relative error is defined as a vector with the same size as the input histogram, and contains one relative error value per bin. More specifically:
\begin{align*}
 	\algoname{Rel}(\mathbf{x}, \mathbf{\tilde{x}}) &= \big[ |x_i - \tilde{x}_i|/\max (x_i, \delta ) \big]_{i\in [d]} 
\end{align*}
We report the median ($\algoname{Rel}_{50}$) and 95\% ($\algoname{Rel}_{95}$) of the per-bin relative error to capture the average and worst-case empirical per-bin errors. For our experiments we ran each algorithm $10$ times on  each input and estimate the value of the loss function by taking the average across the executions, we also use $\delta = 1.0$.

%!TEX root = one_sided_privacy.tex

\subsection{Experimental Results}
\label{sec:exp-results}

\subsubsection{Classification}
\label{sec:exp-results-classification}

\noindent \textbf{Algorithms:} We compare four algorithms.
\begin{enumerate*}[label=\arabic*)]
	\item \emph{All NS} is a {\em baseline that does not satisfy OSDP} that runs a non-private classifier on all the non-sensitive records. This captures the threshold algorithm that satisfies personalized DP \cite{jorgensen:icde15}. Recall that personalized DP is susceptible to the exclusion attack. 
	\item \emph{\algoname{OsdpRR}} is our OSDP strategy of releasing true non-sensitive records using \algoname{OsdpRR} and then running a non-private classifier. 
	\item \emph{ObjDP} is an implementation of the objective perturbation technique for differentially private empirical risk minimization \cite{chaudhuri2011differentially}. As prescribed by the authors, we normalized feature vectors to ensure the norm is bounded by $1$ before running the method on our dataset. 
	\item \emph{Random} is a baseline that randomly predicts a label based on just the label distribution (and ignoring all the features).
\end{enumerate*}

\begin{figure}[htp]
	\centering
	\begin{subfigure}{0.48\columnwidth}
		\includegraphics[width=\columnwidth]{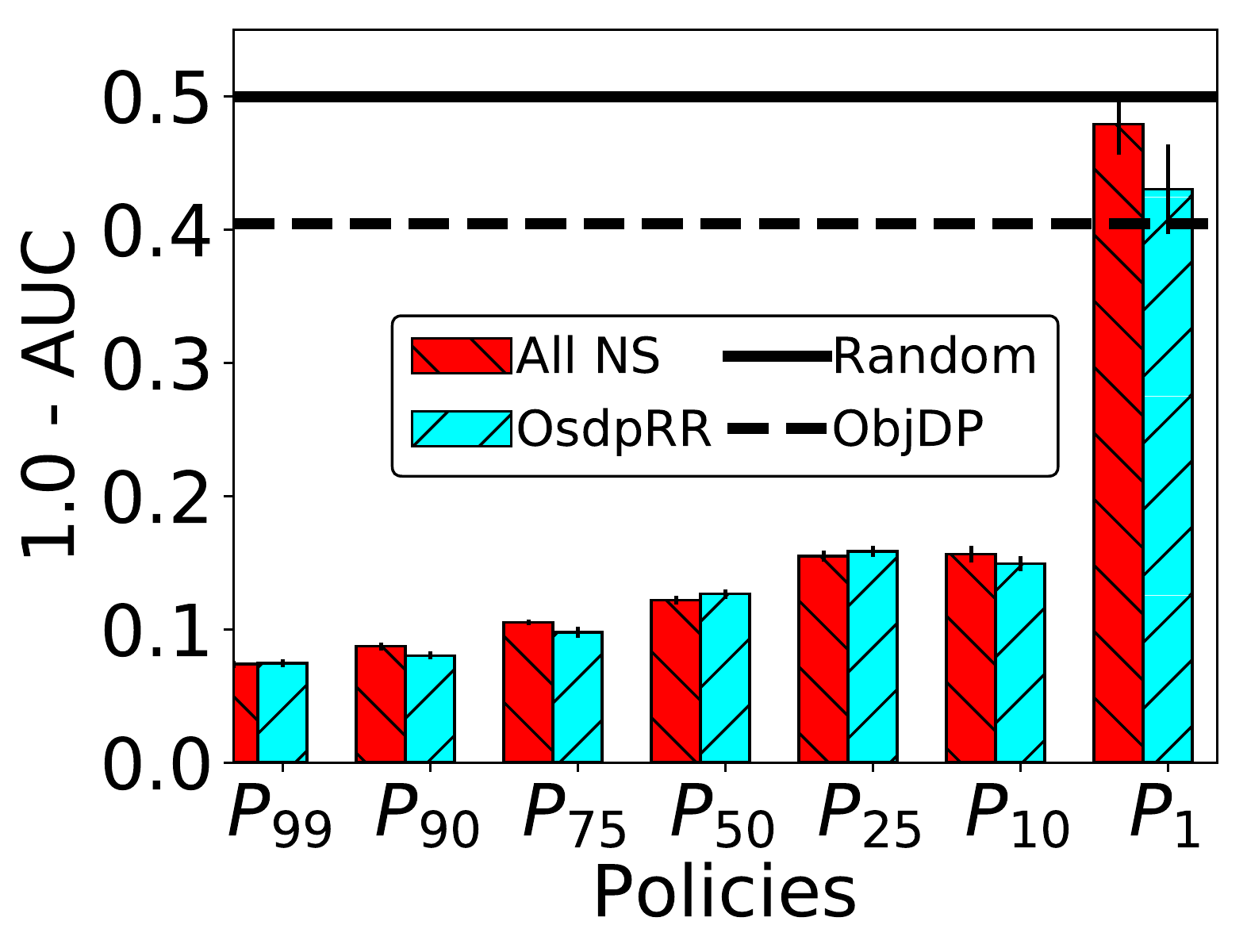}			
		\caption{$\epsilon=1$}	
	\end{subfigure}
	~
	\begin{subfigure}{0.48\columnwidth}
		\includegraphics[width=\columnwidth]{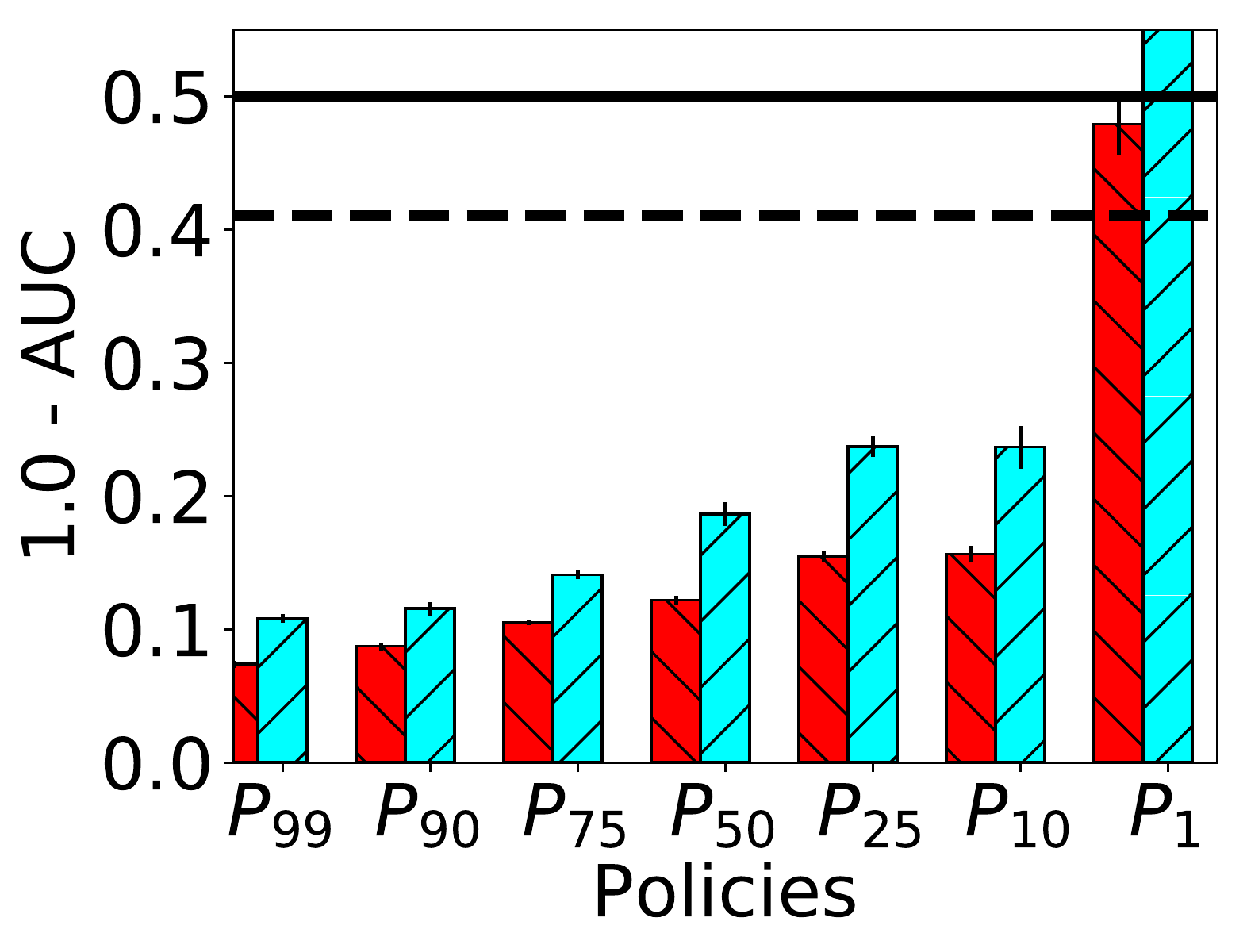}		
		\caption{$\epsilon=0.01$}	
	\end{subfigure}
		\caption{\hl{AUC metric for residents classification.}}
		\label{fig:clas_tipper_auc}
\end{figure}

\noindent  \textbf{Results:} In Figure~\ref{fig:clas_tipper_auc}, we report the error ($1-$AUC) for the policies described in Section \ref{sec:data-tip} for two levels of privacy budget, $\epsilon = 1.0$ and $\epsilon = 0.01$. In each figure, we plot All NS and \rr ~using bars, as their accuracy differs based on the fraction of non-sensitive records in the dataset. We report the error of ObjDP and Random using a straight line.  We see that \rr achieves as low an error as the non-private baseline All NS, and the absolute errors are close to 10\%. As expected, as the fraction of non-sensitive records decreases the error increases (as the number of records used by both All NS and \rr to train the LR model decreases). ObjDP has high error close to that of the Random baseline. This shows that treating all records as sensitive results in a high error and OSDP algorithms alleviate this problem by leveraging on non-sensitive records.

\subsubsection{High Dimensional Histograms}\label{sec:exp-results-high}

\noindent \textbf{Algorithms:}
As noted earlier, under DP the sensitivity of an n-gram  histogram is $64^n$. for that reason, we use \emph{truncation} \cite{KNRS13:analyzing} and choose at most $k$ n-grams from each daily trajectory. This step reduces the sensitivity of releasing the n-gram to $2\cdot k$.  We report the MRE of the Laplace mechanism with truncation of $k=1$ (denoted by LM T1).  We also report the MRE of Laplace mechanism with the optimal choice of $k$ for truncation (denoted by LM T*), 
 this approach does not satisfy DP.  Given the high domain size,  it is intractable to materialize and perturb the complete histogram for LM  when $n$ is large. In our experiments, we only materialize  non-zero counts, so for the MRE computation we analytically computed the contribution of the zero counts to it.

% We used the following formula to estimate the total MRE.
% \begin{align*}
% 	\algoname{MRE} (\mathbf{x}, \mathbf{\tilde{x}})  = \frac{1}{d}\displaystyle\sum_{i: x_i > 0}  |x_i - \tilde{x}_i| / \max(x_i, \delta) +
% 	\frac{1}{d}\displaystyle\sum_{i: x_i = 0}  \frac{1}{\epsilon \cdot \delta}
% \end{align*}
% where $1 / \epsilon$, is the expected $L_1$ error when the true frequency is zero.
% This problem does not arise in case of $\rr$, because it always returns a subset of the non-zero n-grams. 

\begin{figure}[h]
	\centering
	\begin{subfigure}{0.48\columnwidth}
		\includegraphics[width=\columnwidth]{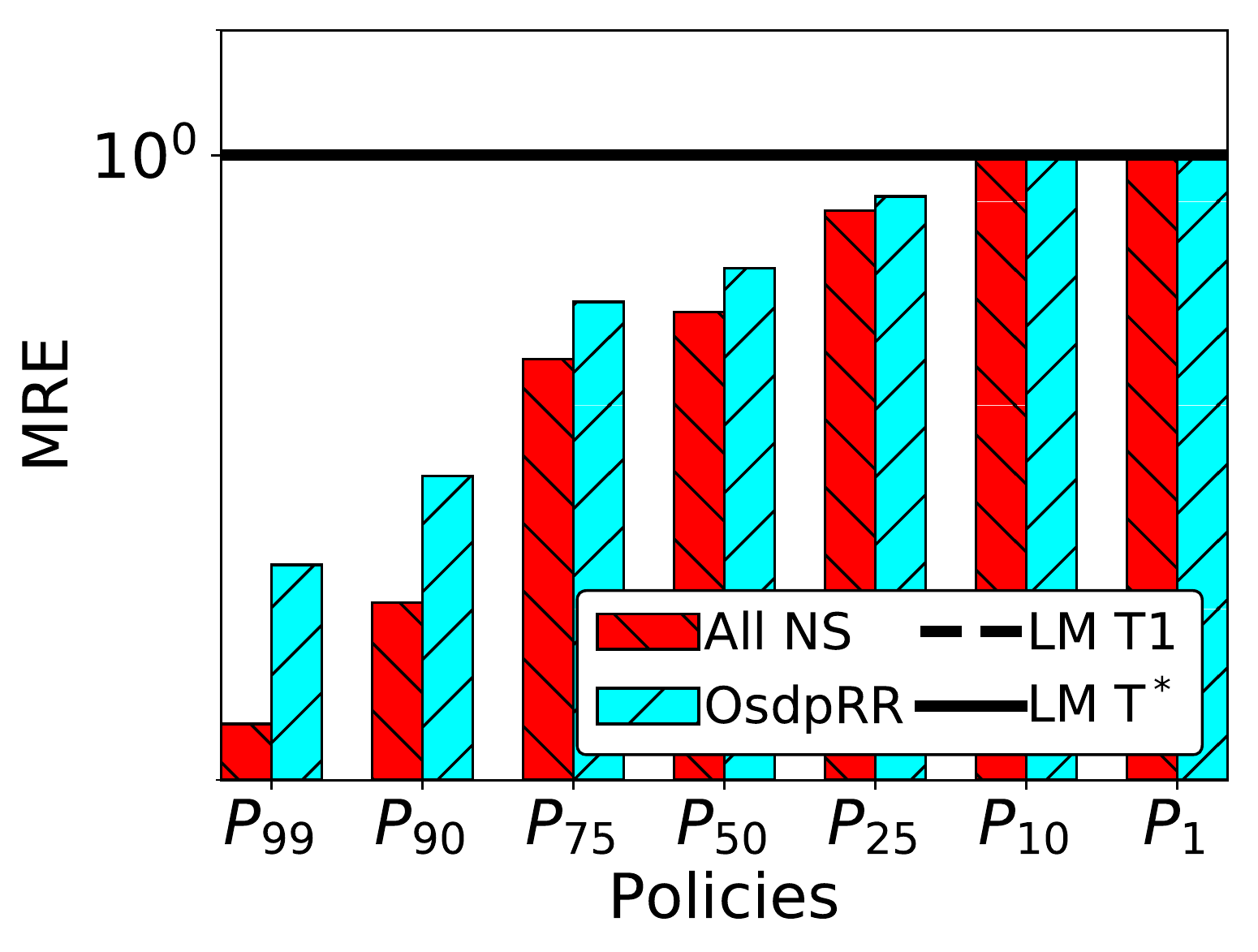}			
		\caption{$\epsilon=1$}	
	\end{subfigure}
	~
	\begin{subfigure}{0.48\columnwidth}
		\includegraphics[width=\columnwidth]{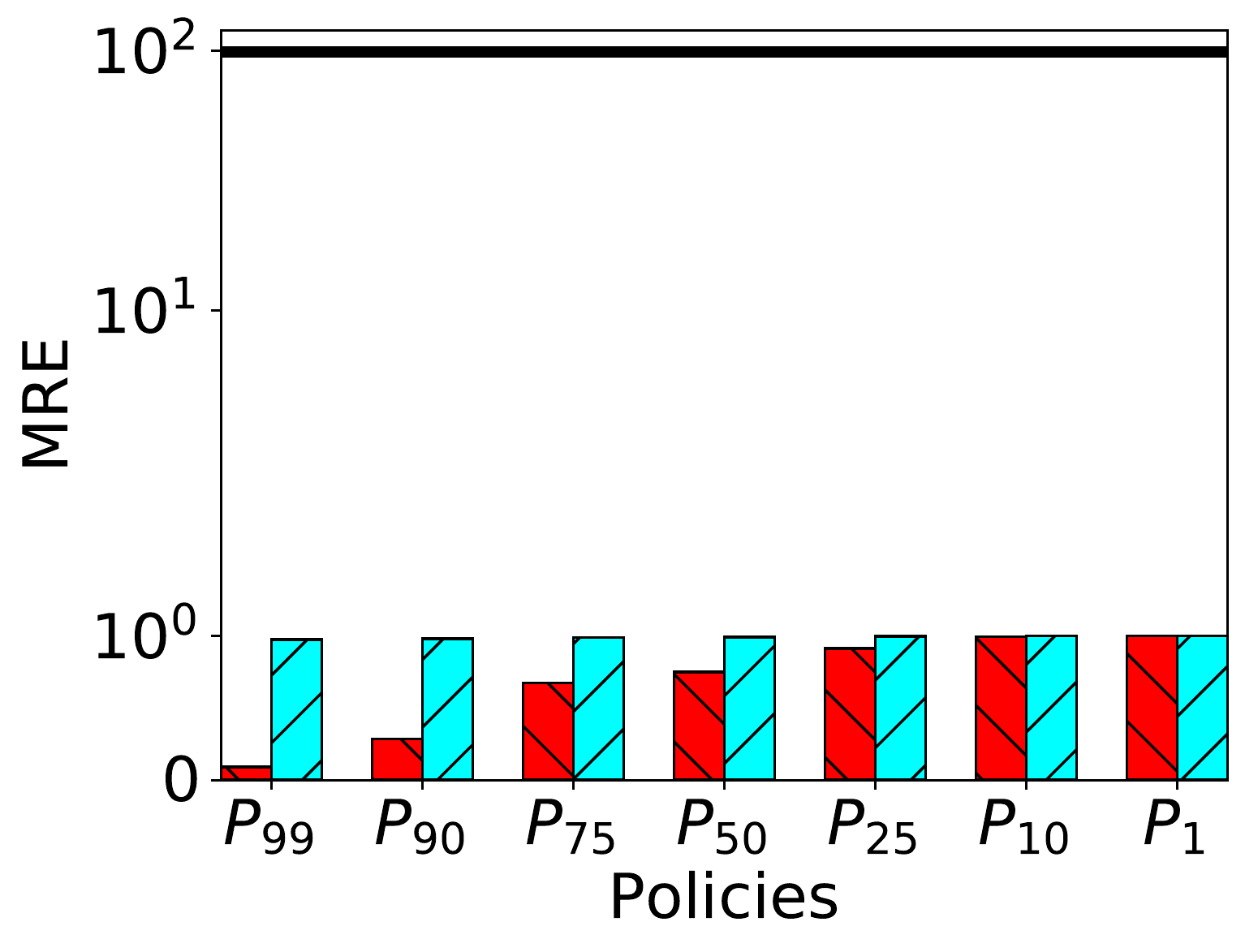}		
		\caption{$\epsilon=0.01$}	
	\end{subfigure}
	\caption{\label{fig:4grams} \hl{Mean Relative Error 4-grams.}}

	\begin{subfigure}{0.48\columnwidth}
		\includegraphics[width=\columnwidth]{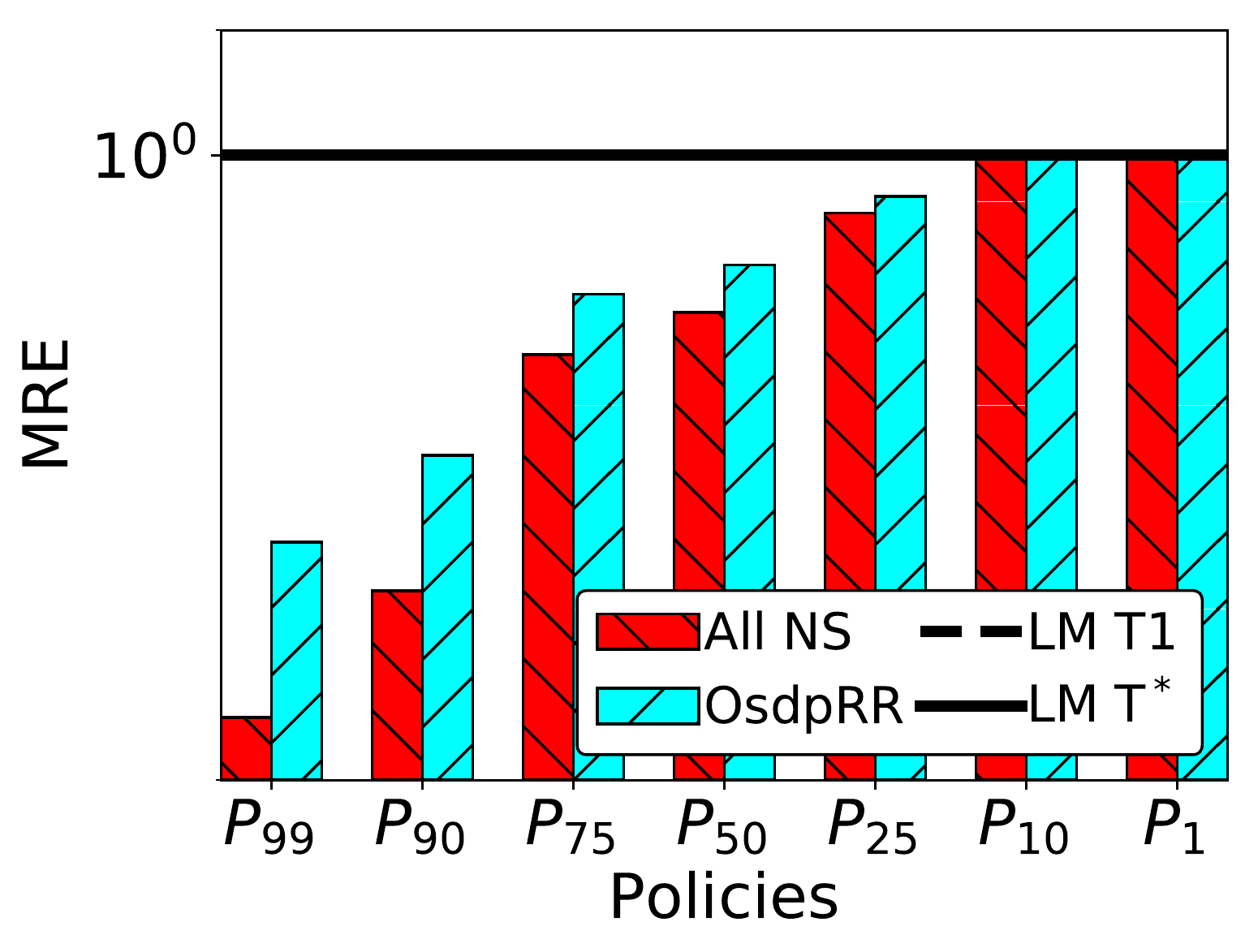}			
		\caption{$\epsilon=1$}	
	\end{subfigure}
	~
	\begin{subfigure}{0.48\columnwidth}
		\includegraphics[width=\columnwidth]{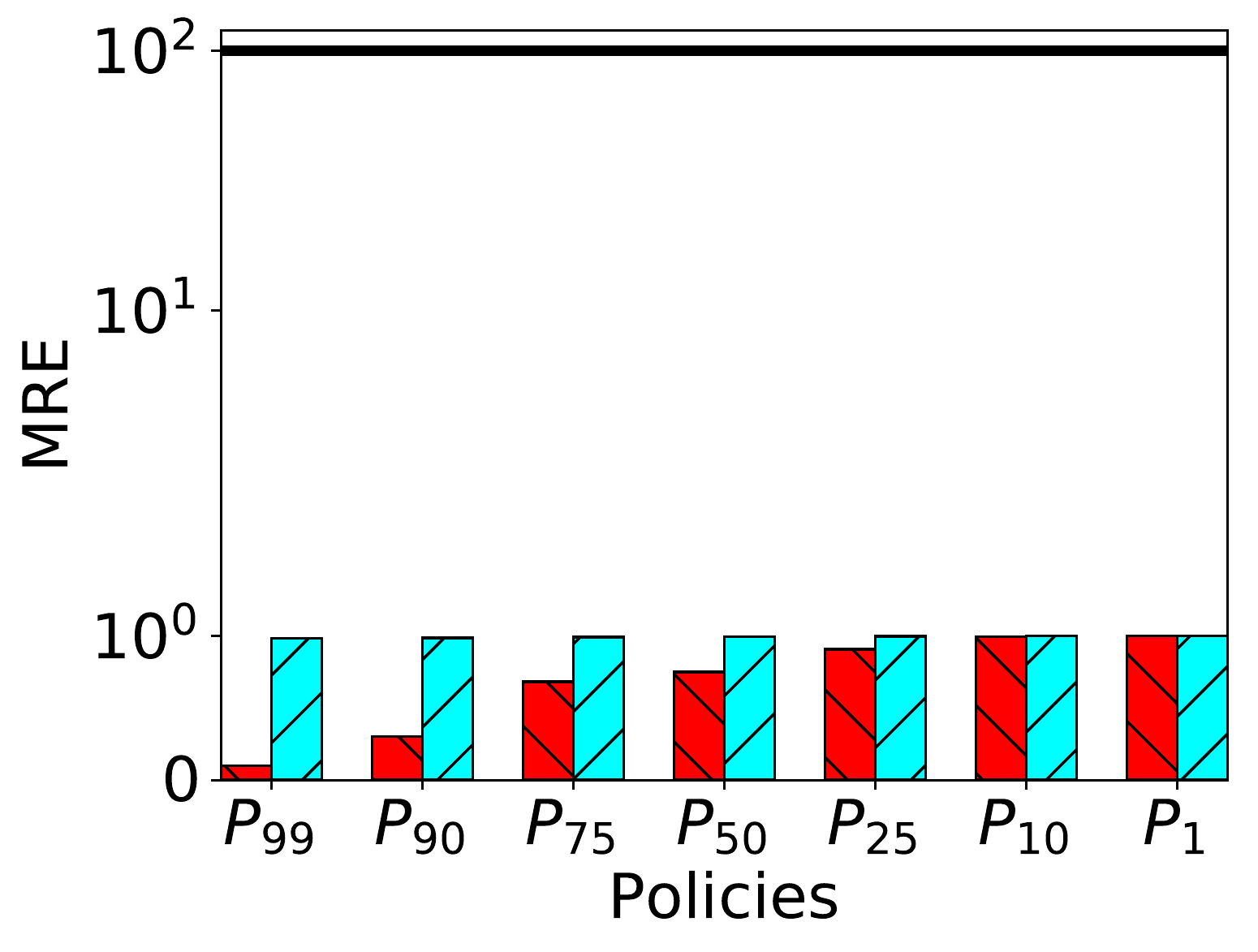}		
		\caption{$\epsilon=0.01$}	
	\end{subfigure}
		\caption{\label{fig:5grams}\hl{Mean Relative Error 5-grams.}}	
\end{figure}

\noindent \textbf{Results:} In Figures \ref{fig:4grams}, \ref{fig:5grams}, we plot the MRE error of the 4-grams and 5-grams for all policies considered.  
We report All NS and \rr error using bars, and the errors of LM T1 and LM T* using straight lines -- since the DP algorithms' error is independent of the policy used. We observed that the optimal choice of truncation parameter $k$ for the 4- and 5-gram counting task was $k=1$. We see that All NS has error smaller than \algoname{OsdpRR}, but the difference is not too large when there are 10-50\% sensitive records in the dataset. 
When $\epsilon = 1$, LM algorithms has error comparable to $\rr$, when $50\%$ of the records are non-sensitive. 
For $\epsilon = 0.01$, the LM algorithms have an order of magnitude poorer performance than the \algoname{OsdpRR} algorithm. 
Thus, we see again that leveraging on the presence of non-sensitive records results in significant utility gains. 
We did not consider more sophisticated DP techniques for releasing high-dimensional histograms. We believe their error will also be an order of magnitude larger than the OSDP algorithm, since they will also have to contend with the high sensitivity and would need to use truncation.

\subsubsection{Low Dimensional Histograms}\label{sec:exp-results-low}
\noindent \textbf{Algorithms:} For the task of releasing low dimensional histograms, we compare a total of $4$ OSDP and $2$ DP algorithms. The DP algorithms are: \algoname{Laplace} mechanism and \algoname{DAWA}, the state of the art DP algorithm for this task according to a recent benchmark study \cite{DPBench}. The OSDP algorithms are: \algoname{OsdpRR}, \algoname{OsdpLaplace}, \algoname{OsdpLaplaceL1} and \algoname{DAWAz}. For \algoname{DAWAz} we used $\rho = 0.1$ fraction of the privacy budget to run \algoname{OsdpRR} to estimate the set of zero counts in the histogram.
%  For clarity, in our figures we show only \algoname{DAWA}, \algoname{OsdpLaplaceL1} (Algorithm~\ref{algo:osp_laplace_l1}), and \algoname{DAWAz} (Algorithm~\ref{algo:dawaz}) -- since these algorithms have the best performance across all settings.

\paragraph{Results: TIPPERS}
\label{sec:exp-results-low-real}

\begin{figure}[htp]
	\centering
	\begin{subfigure}[htp]{0.25\textwidth}
		\centering
		\includegraphics[scale=0.18]{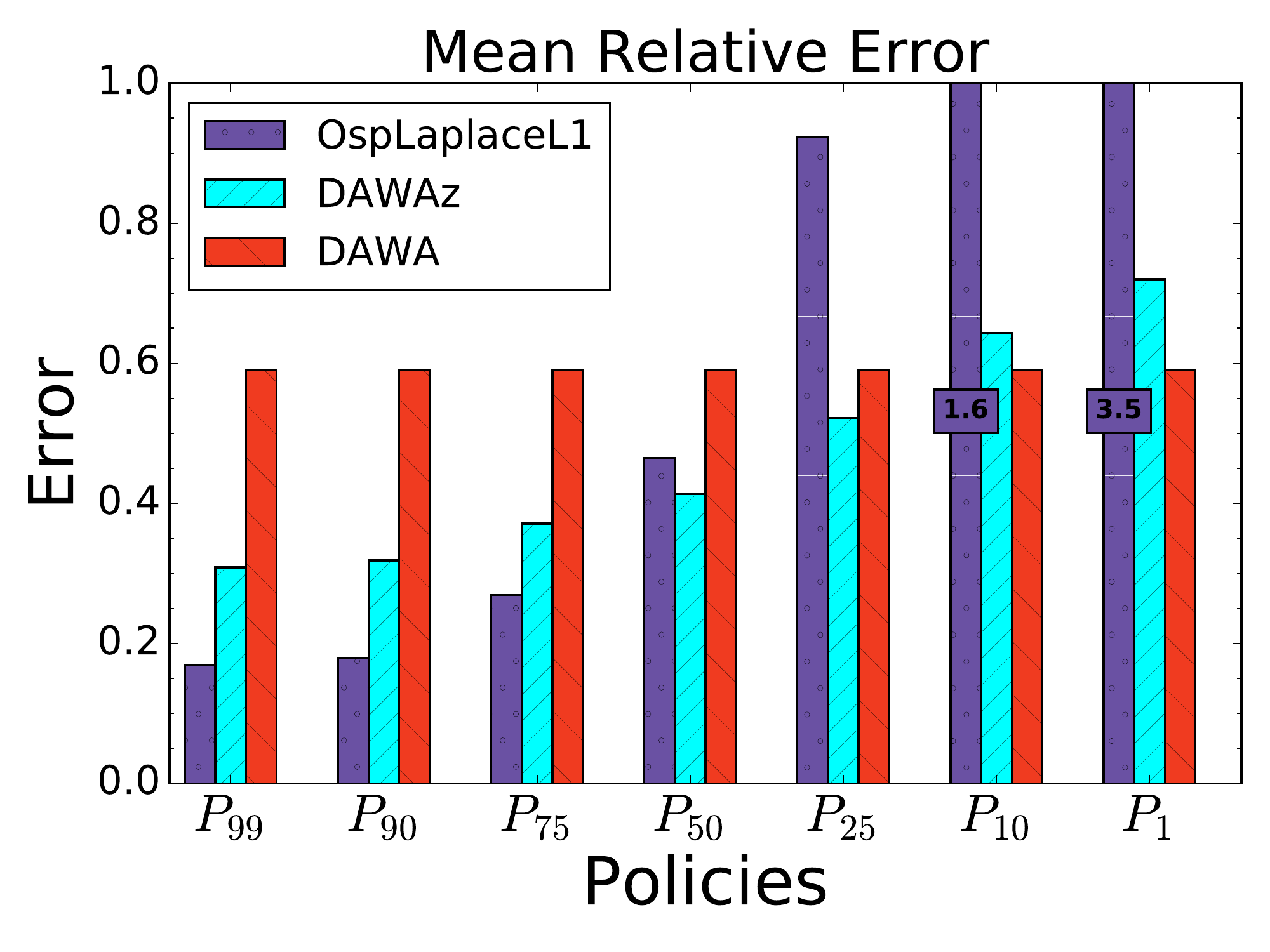}				
				\caption{\label{fig:tippersMRE-1}\hl{ $\epsilon = 1.0$}}
	\end{subfigure}%
	\begin{subfigure}[htp]{0.25\textwidth}
		\centering
		\includegraphics[scale=0.18]{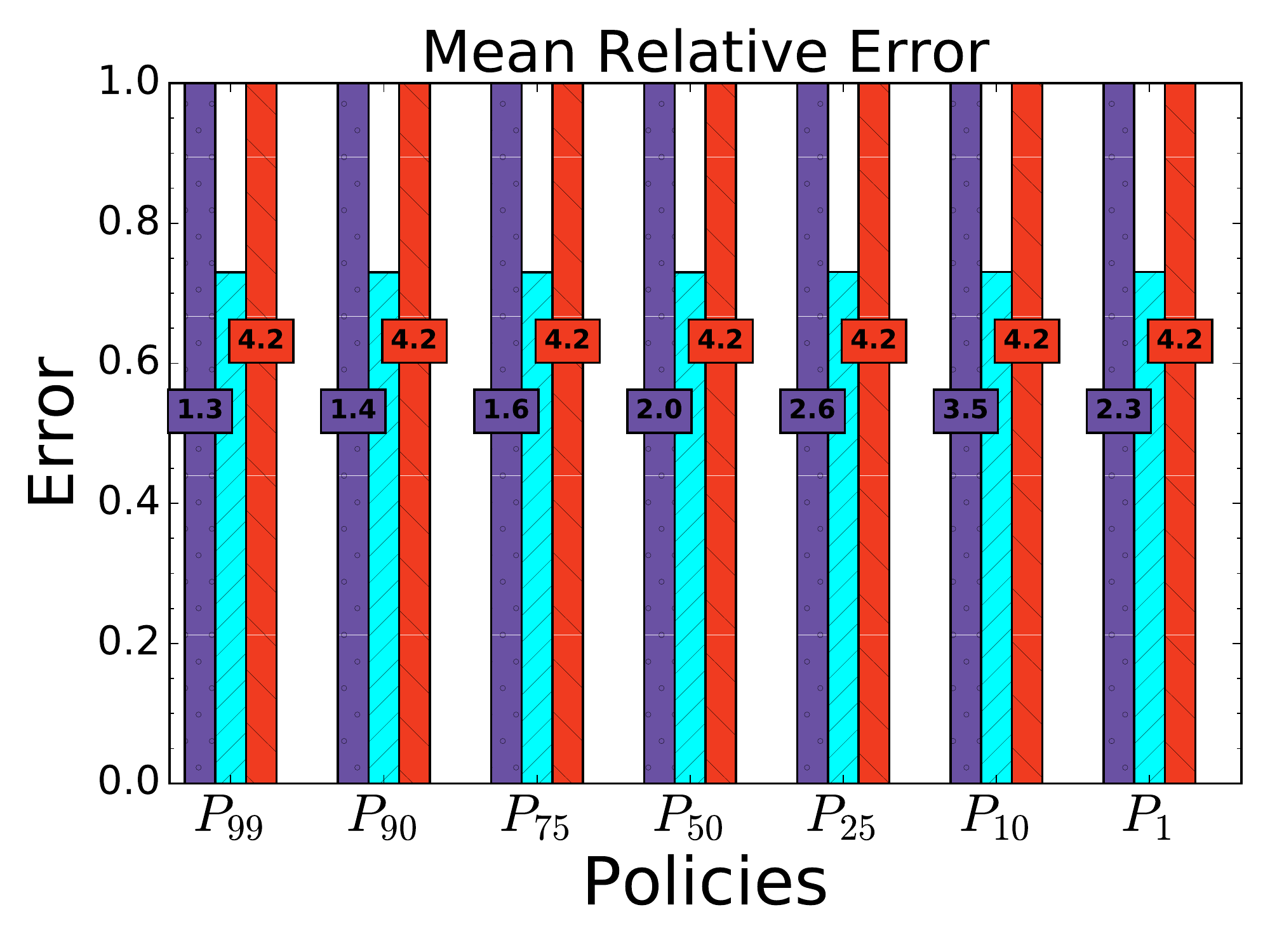}
				\caption{\label{fig:tippersMRE-01} \hl{$\epsilon = 0.01$ }}
	\end{subfigure}%
		\caption{\label{fig:tippersMRE}\hl{ Mean Relative Error on the TIPPERS histogram.}}
\end{figure}

\begin{figure}[htp]
	\centering
	\begin{subfigure}[htp]{0.25\textwidth}
		\centering
		\includegraphics[scale=0.18]{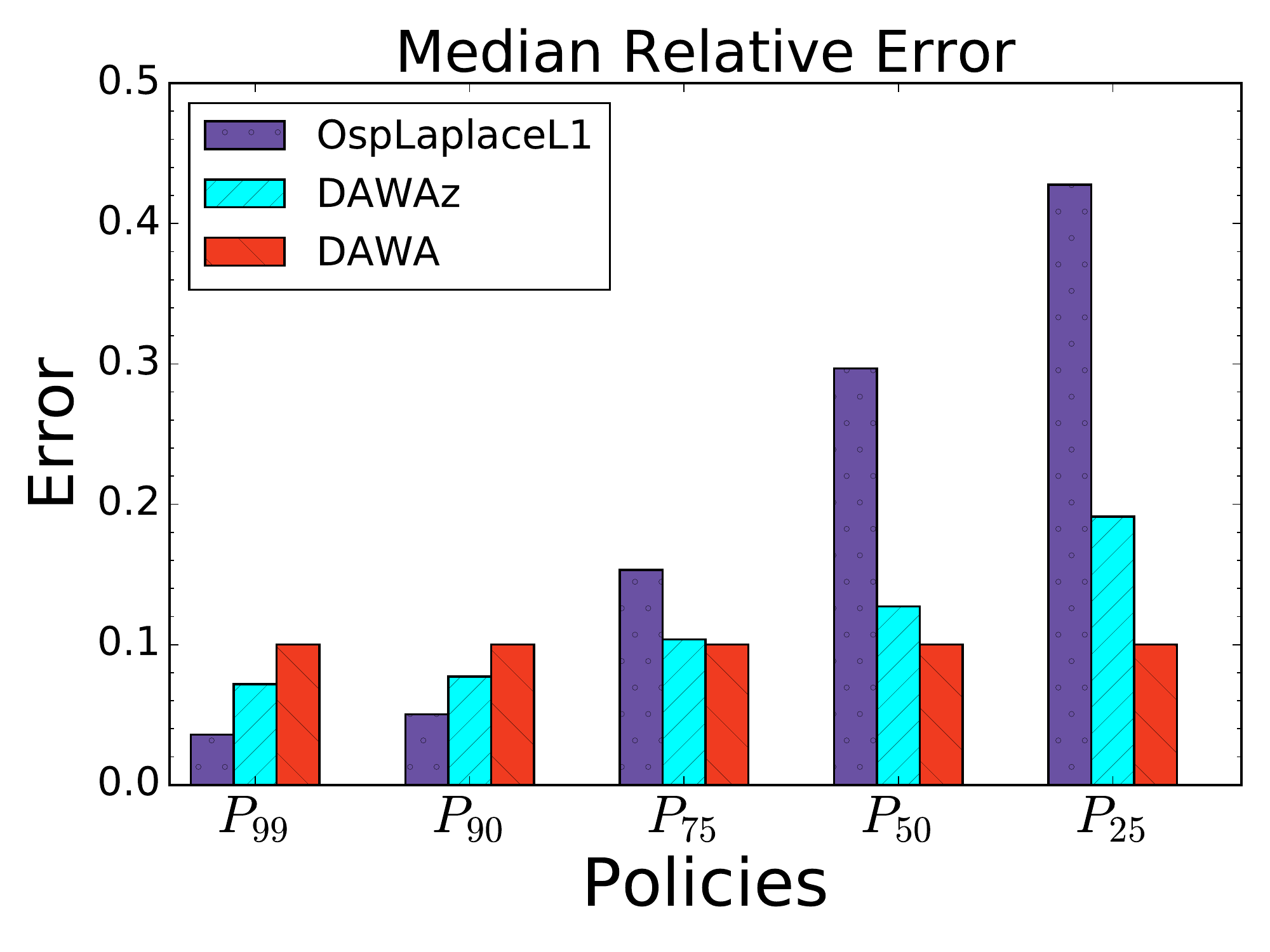}				
				\caption{\label{fig:tipp-rpb50}\hl{ $\algoname{Rel}_{50}$  }}
	\end{subfigure}%
	\begin{subfigure}[htp]{0.25\textwidth}
		\centering
		\includegraphics[scale=0.18]{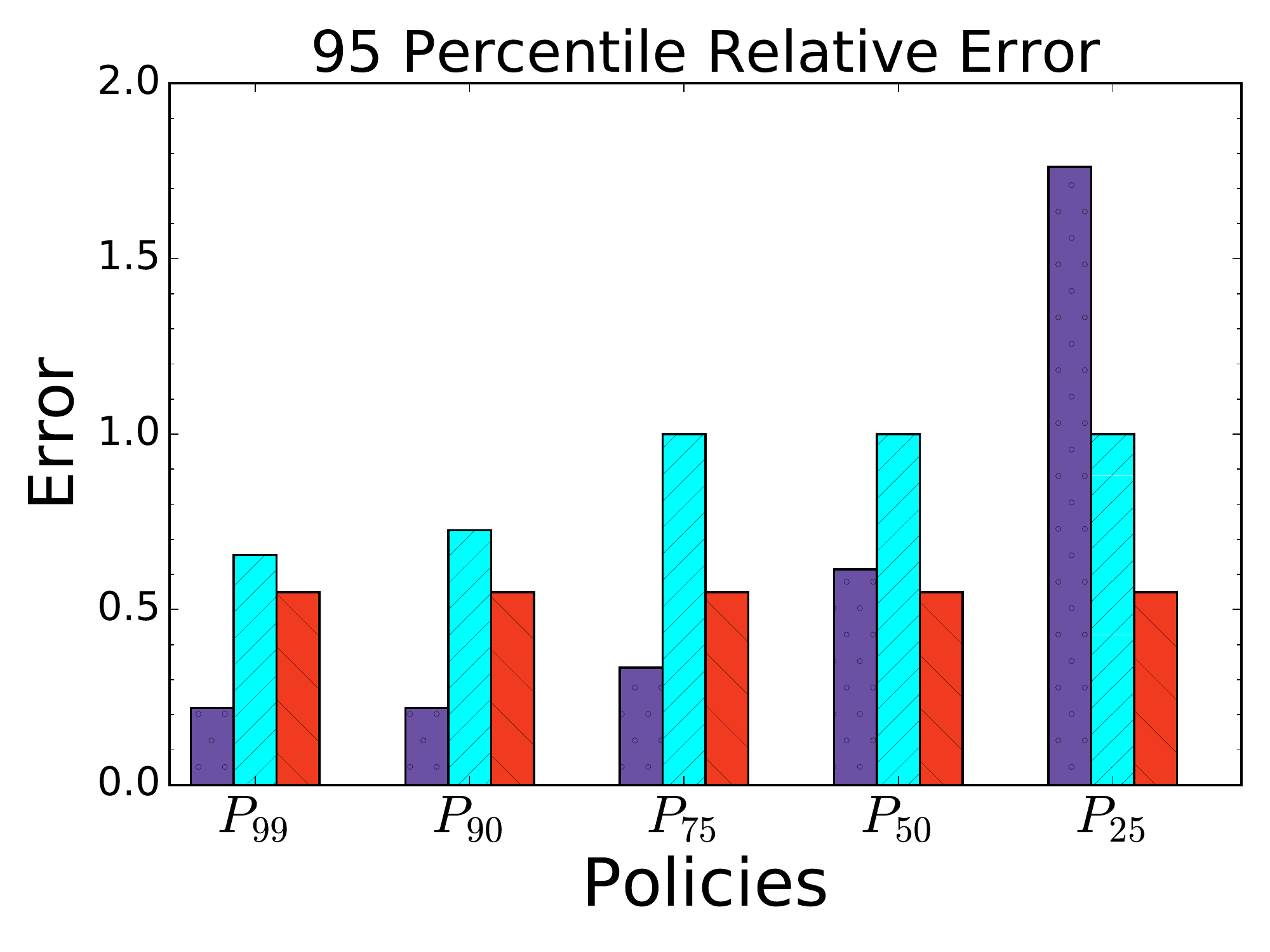}
				\caption{\label{fig:tipp-rpb95} \hl{$\algoname{Rel}_{95}$} }
	\end{subfigure}
		\caption{\label{fig:tipp-rpb} \hl{Relative per Bin Error for the TIPPERS histogram at $\alpha$ percentile, for $\epsilon = 1$. }}
		\label{fig:sparsity}
\end{figure}

\hl{Figure} \ref{fig:tippersMRE}  \hl{reports the mean relative error for two different values of the privacy parameter: $\epsilon \in \{1, ~0.01\} $ across all policies. We see that for $\epsilon = 1$ the OSDP algorithms outperform the DP algorithm for policies with more than $25\%$ non-sensitive records, but for smaller fractions the DP algorithm is clearly more competitive. The insight from this result is that for datasets with mainly sensitive records, using a DP algorithm is preferable. In} \cref{fig:tippersMRE-01} \hl{we report our findings for $\epsilon = 0.01$ where again the OSDP algorithm \textnormal{\textsc{OsdpLaplaceL1}} is outperformed from the DP algorithm at the same $25\%$ mark, but this time  \textnormal{\textsc{DAWAz}} is highly competitive for all policies. This is because the OSDP subroutine of \textsc{DAWAz} over-reports zero bin counts, which is better than adding high scale noise to the true bin counts. For brevity in the rest of our results we report only policies with more than $25\%$ fraction of non-sensitive records and for $\epsilon = 1$.}

In figures \ref{fig:tipp-rpb50}  and \ref{fig:tipp-rpb95} we report the median and 95\%  per bin relative error.  Overall we see that across error metrics, OSDP algorithms offer significant improvements over DAWA for all fractions of non-sensitive values. Even when only half the data records are non-sensitive we see that OSDP algorithms achieve better error rates in MRE and $\algoname{REL}_{95}$, which means that the OSDP algorithms have an advantage in reporting counts of bins with otherwise high error. 
Note that \algoname{DAWAz} is not as competitive as we would expect and is outperformed from \algoname{OspLaplaceL1}. The main reason is that the policy in the $\tippers$ dataset is value based. So some histogram bins contain only sensitive records, while others only consist of non-sensitive records. So \algoname{OsdpLaplaceL1} is able to add normal Laplace noise to the sensitive buckets (ensuring DP) and one-sided noise to non-sensitive buckets (ensuring OSDP). The overall algorithm ensures OSDP by sequential composition.

\paragraph{Results: DPBench-1D}
\label{sec:exp-results-low-benchmark}

\begin{figure}[h]
	\centering
	\begin{subfigure}[htp]{0.25\textwidth}
		\centering
		\includegraphics[scale=0.19]{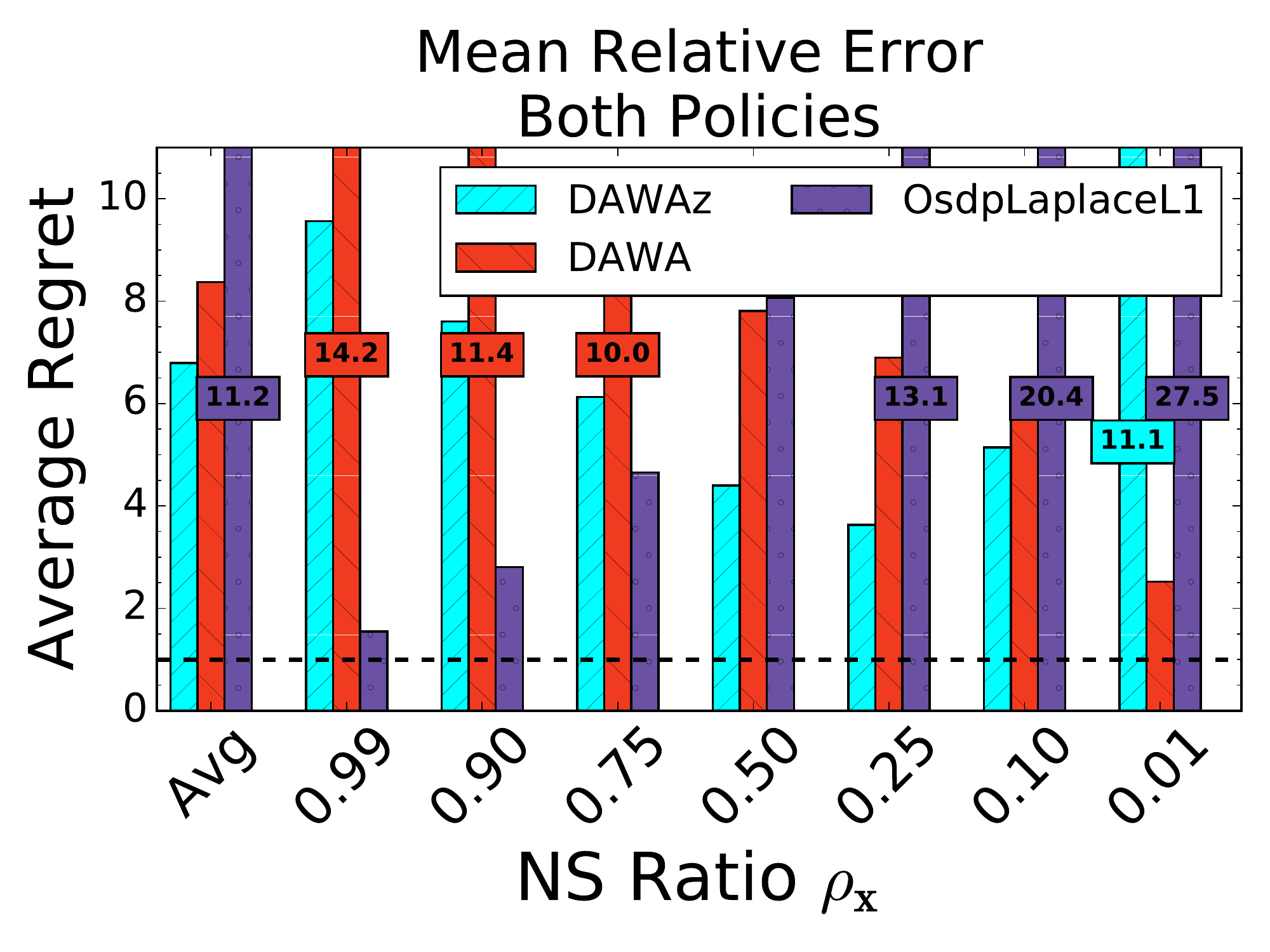}
				\caption{\label{fig:MRE-both-eps1} \hl{$\epsilon = 1.0$}}
	\end{subfigure}%
	\begin{subfigure}[htp]{0.25\textwidth}
		\centering
		\includegraphics[scale=0.19]{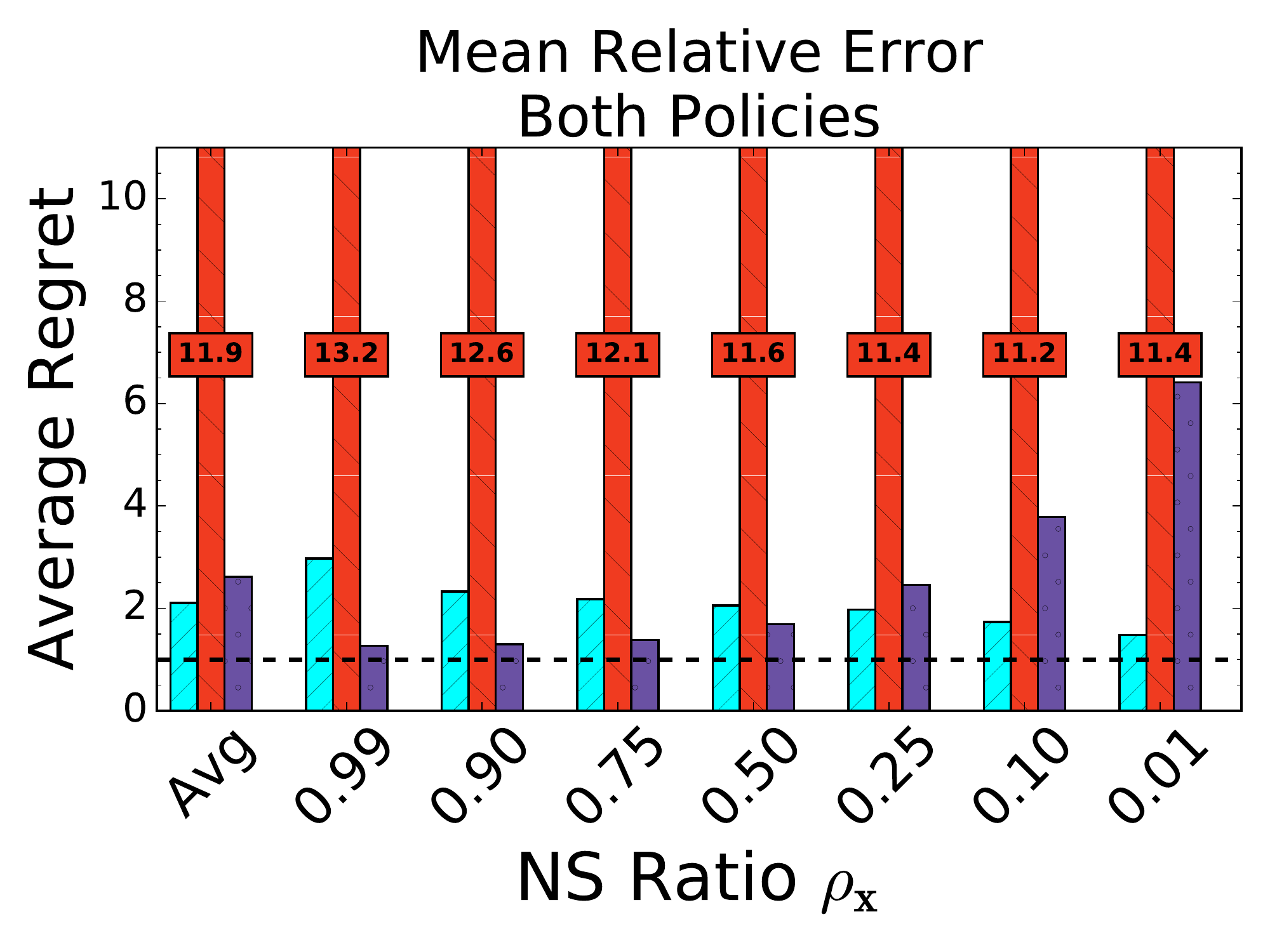}
				\caption{\label{fig:MRE-both-eps01} \hl{$\epsilon = 0.01$} }
	\end{subfigure}
		\caption{\label{fig:MRE-both} \hl{Regret for MRE, across different non-sensitive ratios and different policies.}}
\end{figure}

\begin{figure}[h]
	\centering
	\begin{subfigure}[htp]{0.25\textwidth}
		\centering
		\includegraphics[scale=0.19]{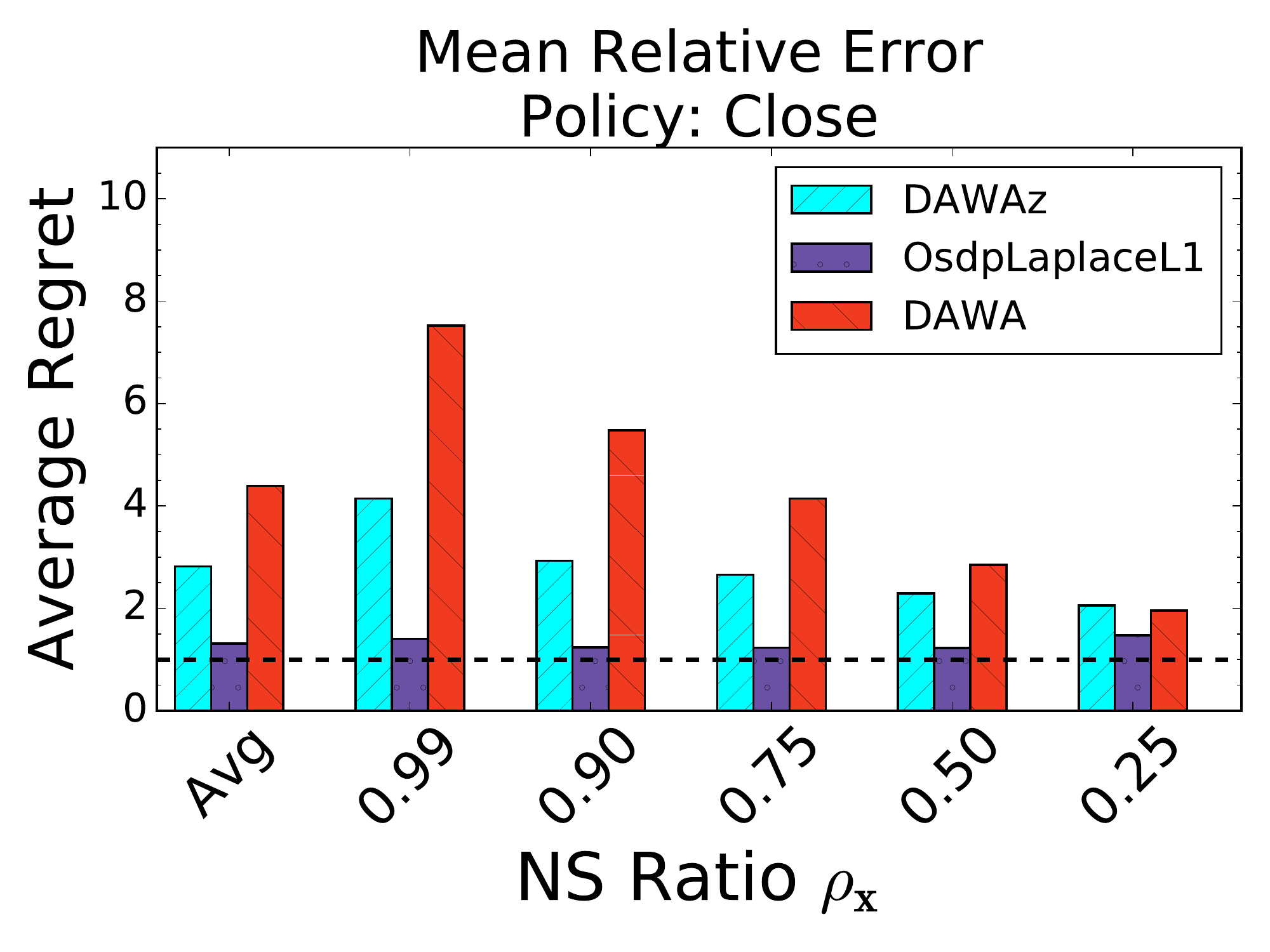}
				\caption{\label{fig:MRE-RP-close} \hl{Close policy.}}
	\end{subfigure}%
	\begin{subfigure}[htp]{0.25\textwidth}
		\centering
		\includegraphics[scale=0.19]{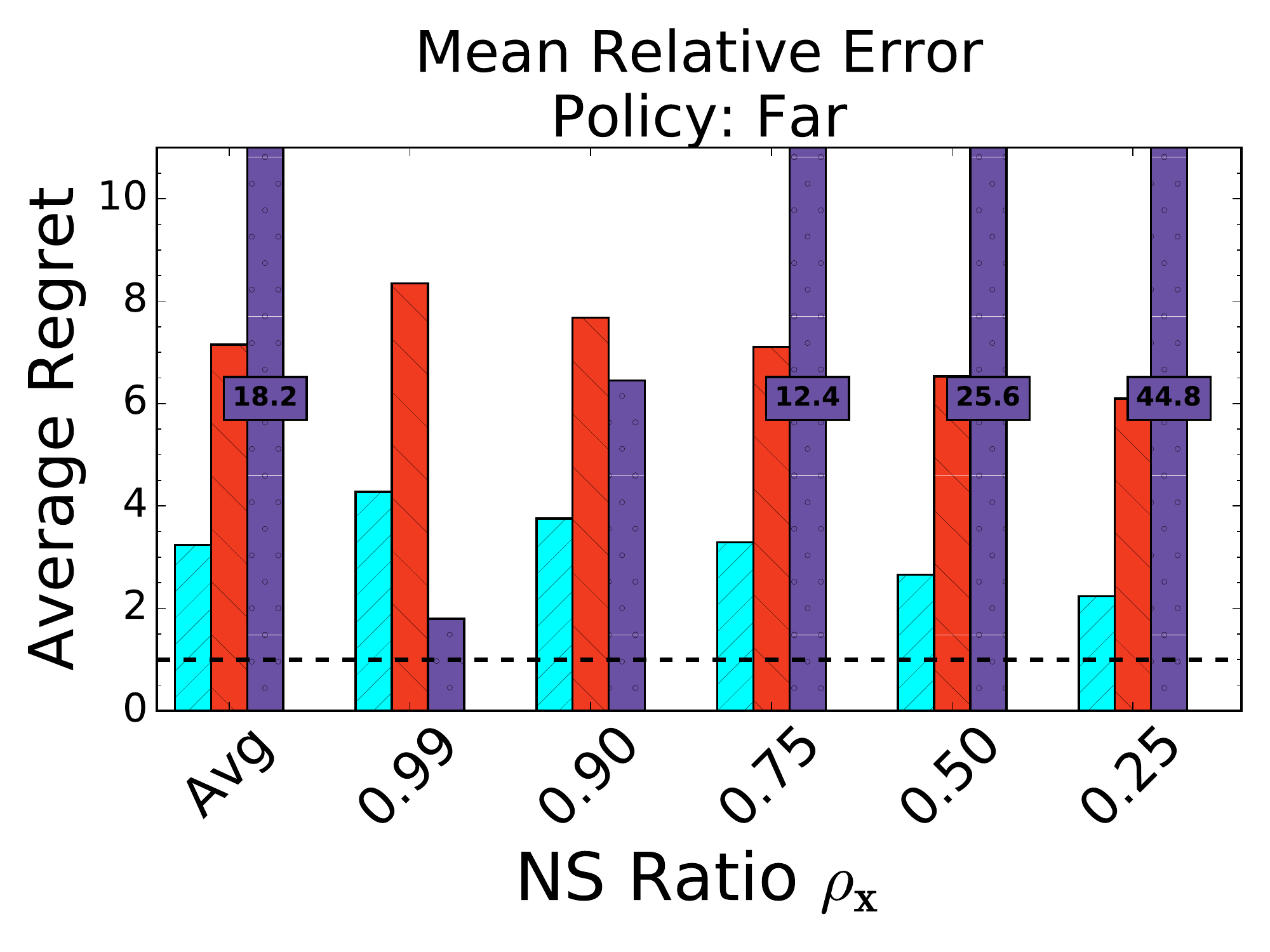}
				\caption{\label{fig:MRE-RP-far} \hl{Far policy.}}
	\end{subfigure}
		\caption{\label{fig:MRE-RP} \hl{Regret for MRE, across different non-sensitive ratios and for fixed policy, for $\epsilon = 1$.}}
		\label{fig:workload}
\end{figure}

\begin{figure}[h]
	\centering
	\begin{subfigure}[htp]{0.25\textwidth}
		\centering
		\includegraphics[scale=0.19]{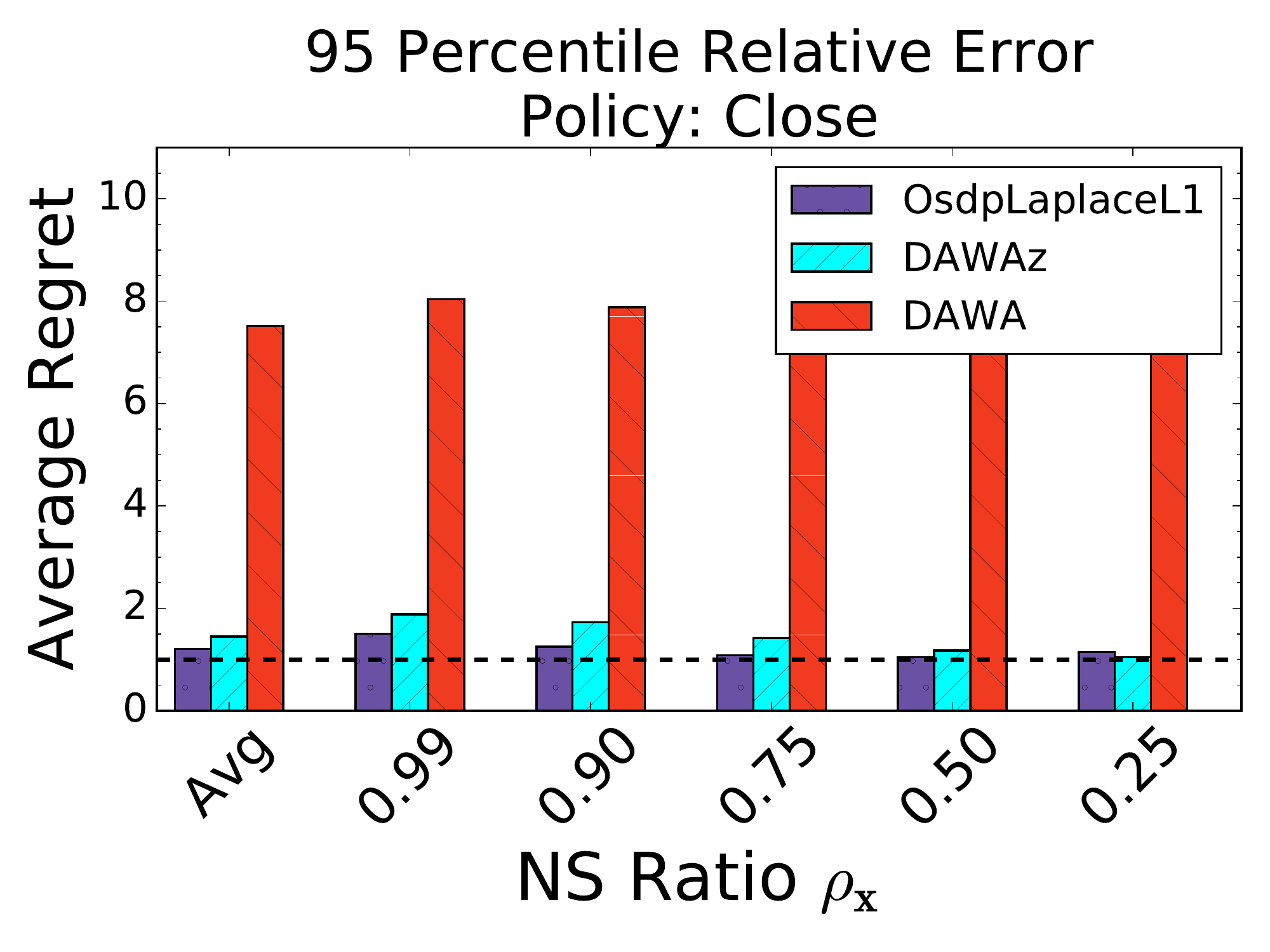}
				\caption{\label{fig:RpB95-RP-close}\hl{ Close policy.}}
	\end{subfigure}%
	\begin{subfigure}[htp]{0.25\textwidth}
		\centering
		\includegraphics[scale=0.19]{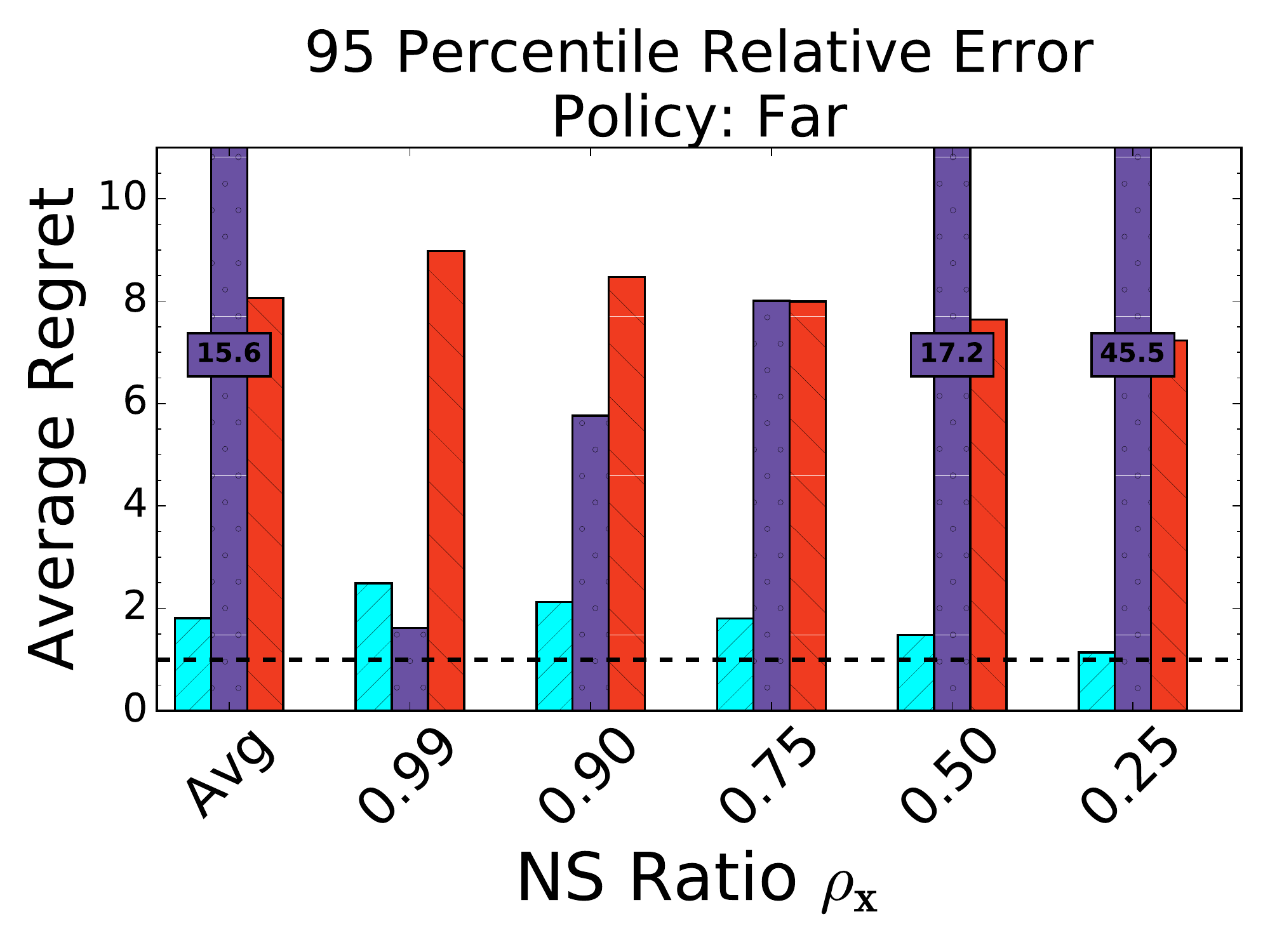}
				\caption{\label{fig:RpB95-RP-far} \hl{Far policy.}}
	\end{subfigure}
		\caption{\label{fig:RpB95-RP}\hl{  Regret for $\algoname{Rel}_{95}$, across different non-sensitive ratios and for fixed policy, for $\epsilon = 1$.}}
		\label{fig:workload}
\end{figure}

\begin{figure}[htp]
	\centering
	\begin{subfigure}[htp]{0.25\textwidth}
		\centering
		\includegraphics[scale=0.19]{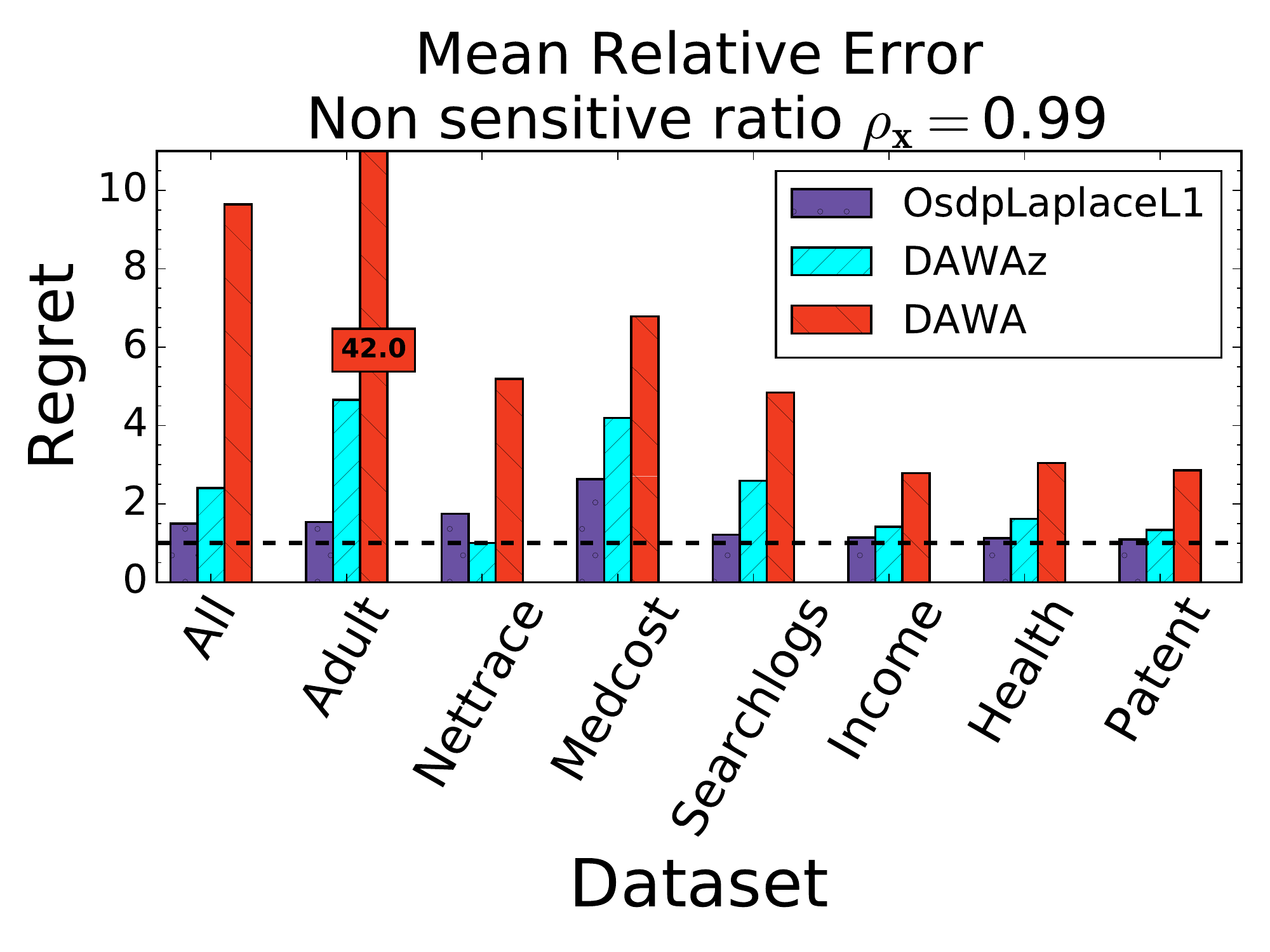}
				\caption{\label{fig:sparse99}$\rho_\mathbf{x} = 0.99$}
	\end{subfigure}%
	\begin{subfigure}[htp]{0.25\textwidth}
		\centering
		\includegraphics[scale=0.19]{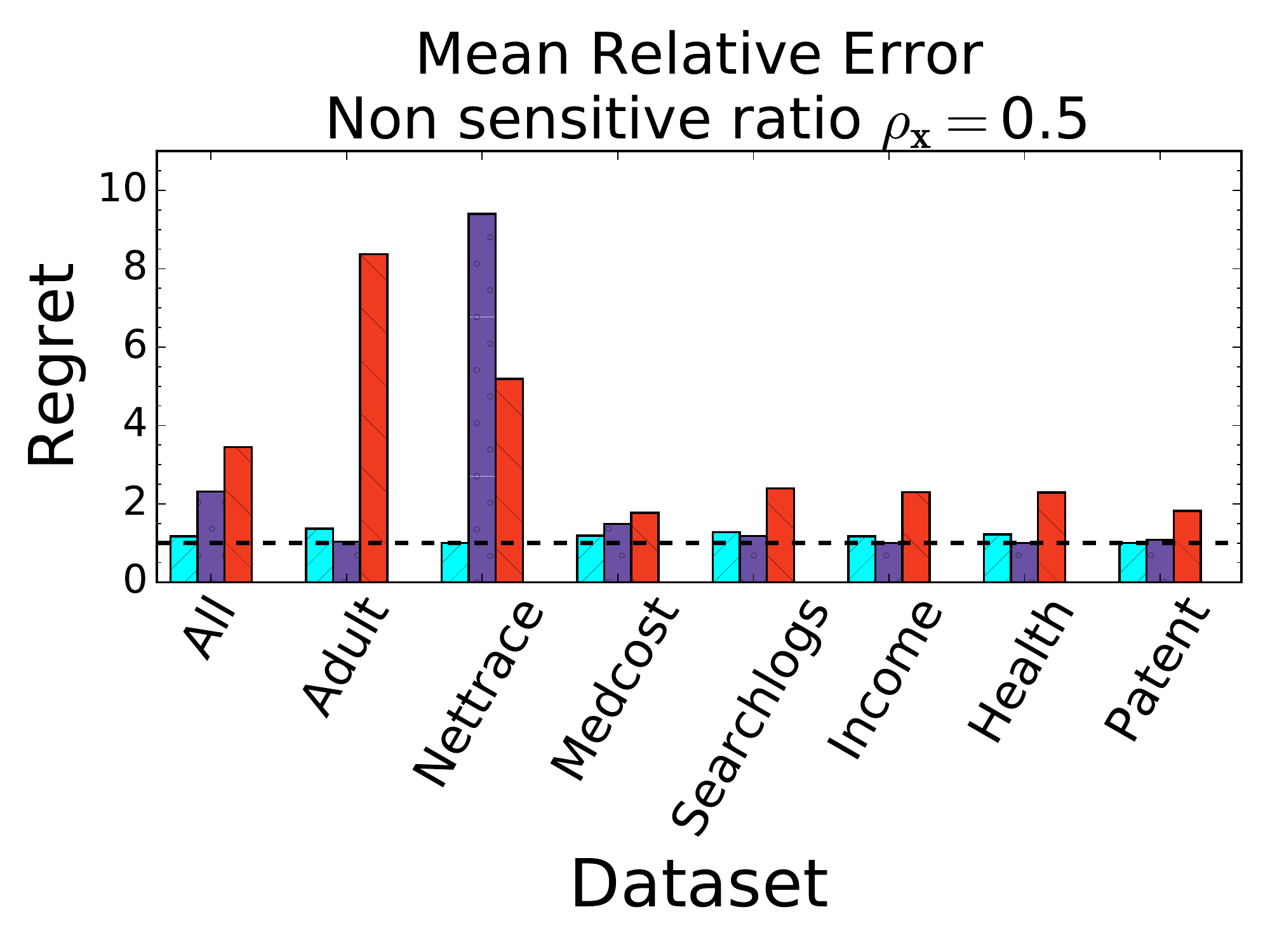}
				\caption{$\rho_\mathbf{x} = 0.50$}
	\end{subfigure}
		\caption{\label{fig:sparse50} Regret for \algoname{MRE} for Close policy and fixed non sensitive ratios.}
		\label{fig:sparsity}
\end{figure}

\begin{figure}[h]
	\centering
	\begin{subfigure}[htp]{0.5\textwidth}
		\centering
		\includegraphics[scale=0.3]{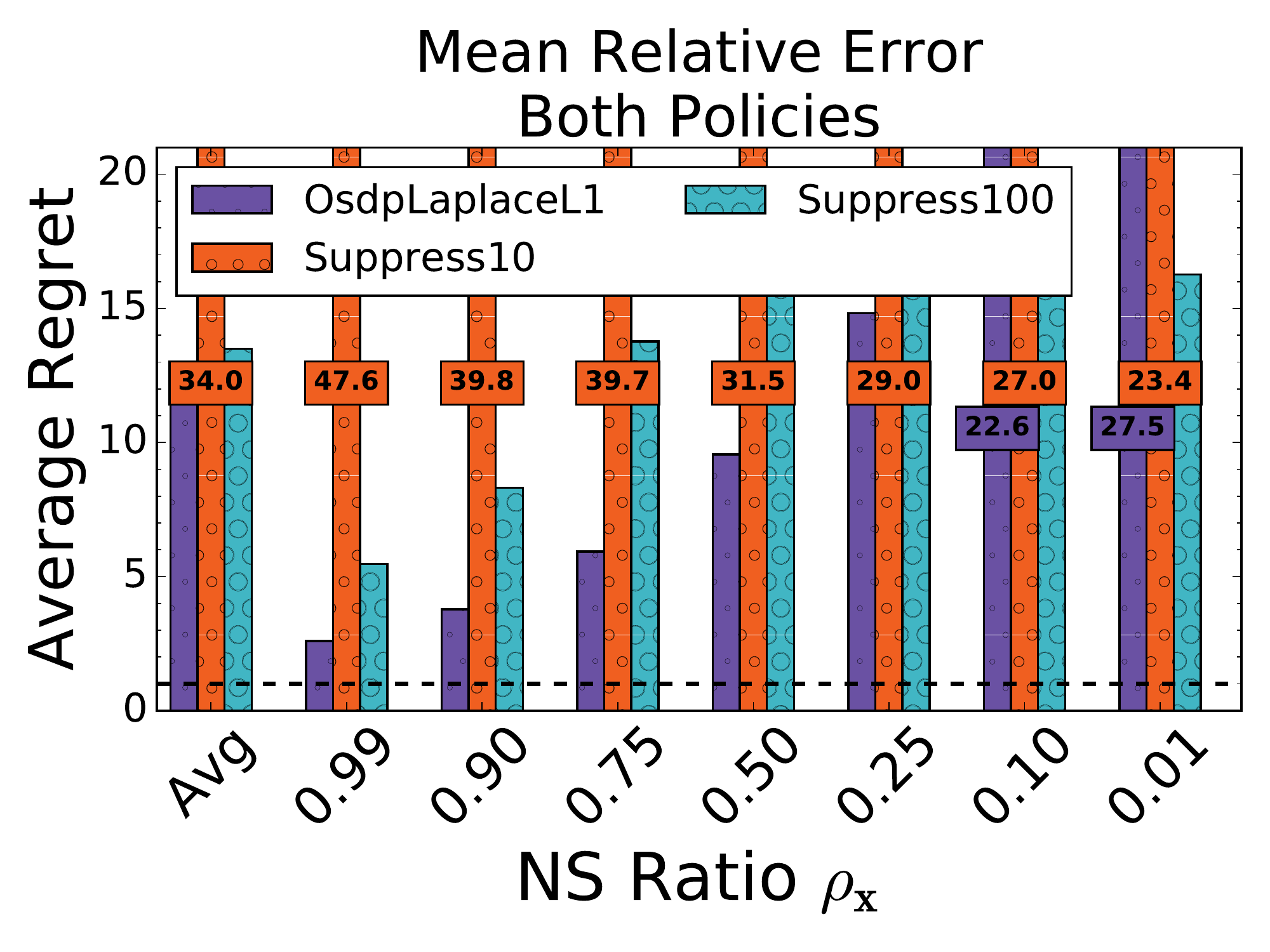}
	\end{subfigure}%
	\caption{\label{fig:MRE-PDP-eps1} \hl{Comparison with PDP algorithms for $\tau =100$ and $\epsilon = 1.0$ measuring the regret of MRE.}}
\end{figure}

In this set of experiments, our goal is to analyze the performance of DP and OSDP algorithms for varying levels of the \textit{closeness} of the non-sensitive histogram and  the \textit{non-sensitive ratio}. Closeness shows how similar $\mathbf{x}$ and $\mathbf{x}_{ns}$ are, and the non-sensitive ratio $\rho_{\mathbf{x}} = \|\mathbf{x}_{ns}\|_1 / \|\mathbf{x}\|_1$ denotes the portion of tuples in $\mathbf{x}_{ns}$. Note that since we report aggregate results of similar inputs (e.g., group inputs by policy),  the aggregated value of the loss function might skew the results due to different error scales between inputs. For this reason, we report  the \textit{regret} of each algorithm w.r.t.  an error measure $\algoname{Err}$ rather than the absolute value of the error. More specifically, for histogram $\mathbf{x}$, budget $\epsilon$, error measure $\algoname{Err}$, and algorithm $\algoname{A} \in \mathcal{A}$ the regret of $\algoname{A}$ is:
\begin{align*}
\text{regret}_\algoname{Err}(A, \mathbf{x}, \epsilon) = \frac{\algoname{Err}(\algoname{A}(\mathbf{x}, \epsilon), \mathbf{x})}{OPT_\mathcal{A} (\algoname{Err}, \mathbf{x, \epsilon})}
\end{align*}

\vskip 3pt \noindent {\bf \hl{Main Results: }}\hl{In figure} \ref{fig:MRE-both}\hl{ we present our results comparing performance across different values of the privacy parameter $\epsilon$. Each figure shows results for both policies, across different values of the non sensitive ratio $\rho_\mathbf{x}$. Note that y-axis reports the average regret for MRE of each algorithm w.r.t. a set of 4 OSDP and 2 DP algorithms, which are averaged over the $7$ histograms from DPBench-1D. We show results only for the most competitive algorithms. The dotted line corresponds to regret equal to $1$ and  since our results are aggregates, single algorithms will rarely have average regret equal to $1$.  Again we observe that low values of $\epsilon$ (i.e., bigger noise scale) favor the hybrid algorithm \textnormal{\textsc{DAWAz}}. The main takeaway is that for ratios $\rho_{\mathbf{x}} \leq 0.25$ the DP algorithm outperforms the ``pure'' OSDP algorithm \textnormal{\textsc{OsdpLaplaceL1}}. Also note that for both $\epsilon$ values the qualitative analysis of the results does not differ by much.  For the above reasons, in our next results we (a) limit $\rho_\mathbf{x} \geq 0.25$ and (b) focus on a single $\epsilon$ value and present separate results for Close and Far policies.}

Next, in figures \ref{fig:MRE-RP} and \ref{fig:RpB95-RP} we present our main results for the MRE and $\algoname{Rel}_{95}$ error measures respectively. The privacy parameter $\epsilon$ is fixed to $1$ and each figure shows the results for a different policy generator.  The y-axis reports the average regret of each algorithm. For Close policies and across all non sensitive ratios and error metrics, OSDP algorithms always outperform DP algorithms, with the highest improvements in high error bins (i.e., $\algoname{Rel}_{95}$). Even for Far policies DAWAz outperforms the best performing DP algorithm DAWA for all non-sensitive ratios. This confirms our insight behind its design, as DAWAz combines the best of both worlds: it  successfully leverages the non-sensitive records to identify zero cells using OSDP algorithms to enhance a state-of-the-art DP algorithm.

%We now focus on the effect of sparsity in the performance of OSDP algorithms. 

In figure \ref{fig:sparsity} we present results for the Close policy and fixed non-sensitive ratios across datasets. Note that these results are not aggregate and in the y-axis we report the regret for a single histogram with a fixed $\rho_\mathbf{x}$ and a fixed policy.  Each group in the x-axis corresponds to a different dataset and the datasets are sorted with descending sparsity (i.e., Adult is the sparsest dataset). Again, the dotted line corresponds to regret equal to $1$. In some cases an algorithm other than the ones shown in the figure achieves a regret of 1 (e.g., \algoname{OsdpRR} has regret $1$ in Adult), but we do not show them since they have very high regrets for other datasets (and to be consistent with our previous figures). Lastly note that Nettrace is a sorted histogram, which highly favors DAWA explaining its sudden drop in regret. 

In the case of the Close policy, OSDP algorithms offer significant performance improvements over using a traditional DP algorithm, across all datasets.  We can see that for the sparser Adult histogram we achieve an improvement on regret of $25\times$, this is because the OSDP algorithms can correctly identify the bins with 0 counts, and incur 0 error on them. As the sparsity of the datasets decreases we see the gap between DAWA (the best performing DP algorithm) and the OSDP algorithms to become smaller. Lastly we observe that as $\rho_\mathbf{x}$ becomes smaller the advantages of \algoname{DAWAz} become more apparent - since \algoname{DAWAz} uses both the sensitive and the non-sensitive part of the histogram.

\noindent {\bf  \hl{Comparison with PDP:}}
\hl{Finally, in} \cref{fig:MRE-PDP-eps1}\hl{, we compare $\algoname{OsdpLaplaceL1}$ with the PDP algorithm $\algoname{Suppress}$ (described in Section} \ref{sec:osp_connection_with_other}) \hl{with thresholds $\tau = \{10\epsilon, 50\epsilon,\allowbreak 100\epsilon\}$ on the benchmark datasets. In figure} \ref{sec:osp_connection_with_other}\hl{ we show the two extreme cases for $\tau =\{10\epsilon, 100\epsilon\}$ and refer to them $\algoname{Suppress10}$ and $\algoname{Suppress100}$ respectively. We observe that  $\algoname{Suppress}$ starts becoming competitive for values of $\tau \geq 100\epsilon$, but this utility comes at a cost of $100$ times worse protection from exclusion attacks (see Theorems} \ref{thm:osp-freedom} and \ref{thm:pdp}).

\section{Extensions and Future Work}
\label{sec:extentions_future}

One-sided differential privacy (OSDP) presents a new direction for privacy research in that it is the first formal privacy definition that exploits the sensitive/non-sensitive dichotomy of data to allow true data to be released with formal privacy guarantees. 
This paper focused on motivating and defining OSDP and highlighting its advantages both in the context of data sharing tasks where differential privacy (DP) either does not apply or offers very low utility, as well as, in
situations where DP performs well, but OSDP appropriately combined with DP offers significant improvement. Below we highlight some of the key extensions that we believe offer interesting  directions of future work. 

\noindent
%{\bf Levels of privacy protection.}
%One interesting direction is to extend the policies to classify data into not just sensitive and non-sensitive, but into multiple levels of sensitivity. For instance, under our current definition records that have $P(r) = 0$ require privacy at the level of $\epsilon$, and $P(r) = 1$ do not require privacy ($\epsilon = \infty$). 
%One could consider a scenario with non-sensitive records have a finite but larger $\epsilon' > \epsilon$. 
%A more general definition would be to define $\pol$ with a larger range; i.e., $P(r)$ could take values in $\{0, 1, \ldots, k\}$, and the definition requires records with $P(r) = i$ have privacy at the level of $\epsilon_i$, where $\epsilon_0 \leq \epsilon_1 \leq \ldots \epsilon_k = \infty$. 
%Extending our algorithms to capture these more general one-sided policies is an interesting avenue for future work. 

%\subsection{Multiple Policies}
%Another extension is one where multiple policies apply to the data. We already encountered this problem when we discussed the composition results in Section~\ref{sec:theory}.If two mechanisms assume different policies $(P_1, \epsilon_1)$ and $(P_2, \epsilon_2)$ on the same dataset, the joint release of both these mechanisms induces a third policy $(\min (P_1, P_2), \epsilon_1 + \epsilon_2)$. \emph{Is there anything we want to say here? Or is it redundant given the previous section?}

\noindent
{\bf Policies that go beyond records.}
In this paper, we considered policy functions that define whether or not a specific record is non-sensitive. One could extend this to the case where a policy specifies a view over the database that is non-sensitive (like in the work on authorization views \cite{rizvi:sigmod2004}), or a policy that determines that a view is sensitive (negative authorization \cite{bertino:tois1999}). While there has been work on whether or not a query leaks any information about a view \cite{machanavajjhala:pods2006,miklau:sigmod2004}, there is no work that bounds the information to leak (like DP)  about the sensitive records in such a  setting. 

\noindent
{\bf One-sided differential privacy and constraints.}
Theorem~\ref{thm:osp-freedom} showed that OSDP algorithms ensure freedom from exlusion attacks as long as an adversary believes the target record is independent of the rest of the database. As in the field of differentially private data analysis problems arise when the records in the database are correlated.  
For instance,  in a location privacy setting, a specific non-sensitive location may be reachable only through a set of locations that are all sensitive. Revealing the fact that a user was in that location (even using our randomized response algorithm) will reveal the fact that the user was in a sensitive location in the previous time point with certainty, thus resulting in a privacy breach. Designing algorithms for OSDP in the presence of constraints in an important and interesting avenue for future work. 
%We note that modeling correlations and constraints in the context of DP has ...
%{\bf Ashwin: point to similar work if any in DP about this}. 

\noindent
{\bf Policy Specification and Enforcement}
OSDP definition and mechanisms assume the presence of an explicit policy function that partitions data into sensitive nad non-sensitive classes. 
As we  discussed in the introduction, there are several application domains where  explicit privacy policies and/or regulations exist that  naturally lead to such
a classification.  Furthermore,  good practice followed by several  data collectors (e.g., e-commerce sites, search engines, etc.)  empowers
users to opt-in or opt-out. Such a practice as well lends itself directly to explicit policy function. Nonetheless, if we are to use OSDP in general application 
context, mechanisms to specify comprehensive policies that dictate data sensitivity or, better still, learn such policies, perhaps through appropriate machine learning
techniques offers rich fertile area for future research. Finally, our focus in this work has been on appropriate privacy definition and mechanisms to realize OSDP. We have
not focused on efficient approach to implement OSDP, especially in online setting wherein user's may dynamically ask queries. 

\section{Conclusions}
\label{sec:conclusion}

In this work we presented a new privacy definition called one-sided privacy that allows for different levels of privacy guarantees (sensitive or non-sensitive) to be applied to different records in the dataset. The level of protection is determined by the value of a record, and hence,  whether a record is sensitive or non-sensitive must also be hidden. This makes the privacy definition unique. We present the fundamental theory behind the privacy definition, compare it to existing definitions and present input and output perturbation algorithms. We present case studies where one-sided privacy can be used, along with effect of our algorithms on the utility of analyzing the data in these scenarios. We concluded with a discussion of extensions and open questions for future work.

\clearpage
{
\bibliographystyle{abbrv}
\bibliography{ref}
}

%%% Old sections %%%
% \input{exp_dp_low_utility}
% \input{appendix_exp_dp_does_not_apply}
% \input{secondary_appendix}
% Use only for the arxiv version

\section{Appendix}
\subsection{Extended OSDP}
We propose a  strict generalization of the definition of OSDP that supports neighboring datasets with different number of records and under which we can show that the parallel composition holds.

\begin{definition}[Extended One-sided $\pol$-Neighbors]
\label{def:extended-osdp-neighbors}
Let $D$, $D'$ be two databases and $\pol$ a policy function. 
$D'$ is an extended one-sided $P$-neighbor of $D$, i.e., $D' \in \eospN{D}$, if and only if $\exists r \in D$ s.t., $P(r) = 0$, such that: $D' = D - \{r\}$ or  $D' = D \cup \{r'\}$, where $r'\in \dom$ and $r' \neq r$.
\end{definition}
% Note that the relation $\eospN{}$ is symmetric.

\begin{definition}[Extended One-sided Differential Privacy]
\label{def:extended-osdp}
	An algorithm $\mech$ satisfies $(\pol, \epsilon)$-Extended one-sided differential privacy (eOSDP), if and only if $\forall \out \subseteq range(\mech)$, and $\forall D, D' 
  \in \eospN{D}$:
  \begin{align*}
    \Pr [\mech(D) \in \out] \leq e^{\epsilon} \Pr [\mech(D') \in \out]
  \end{align*}
\end{definition}

\begin{theorem}[Extended OSDP Equivalence]
\label{thm:ext-equiv}
Any mechanism $\mech$, that satisfies $(\pol, \epsilon)$-eOSDP also satisfies $(\pol, 2\epsilon)$-OSDP.
\end{theorem}

\begin{proof}
Given a database $D$ and $r \in D$ with $P(r) = 0$, let $D' =  D - \{r\} \cup \{r'\}$ an OSDP neighbor of $D$, i.e., $D'\in \ospN{D}$. Then database $D'' = D \cup {r}$ is such that: (a) $D'' \in \eospN{D}$ and (b) $D' \in \eospN{D''}$. Therefore, for mechanism $\mech$ that satisfies $(\pol, \epsilon)$-eOSDP and it holds that $\forall \out \subseteq range(\mech)$:
\begin{align*}
Pr[\mech(D) \in \out ] &\leq e^\epsilon Pr[\mech(D'') \in \out ] \leq e^\epsilon Pr[\mech(D') \in \out ] \implies \\
Pr[\mech(D) \in \out ] &\leq e^{2\epsilon} Pr[\mech(D') \in \out ]
\end{align*}
\end{proof}

Note that the inverse statement of Theorem \ref{thm:ext-equiv} does not necessarily hold, i.e., $(\pol, 2\epsilon)$-OSDP does not imply $(\pol, \epsilon)$-eOSDP.

\begin{theorem}[Parallel Composition]
  \label{thm:osp_parallel_database}
Let $D$ be a database  and $\mathcal{D} = \{ D_1, \cdots, D_n \}$ a partitioning of it. 
Let $\mech_1$, \dots, $\mech_n$ be $(\pol_i, \epsilon_i)$-eOSDP mechanisms and $\polMinRelax$ be the minimum relaxation of  $\pol_1$, \ldots, $\pol_n$. 
Then the composition  $\mech_1 (D_i) \circ \ldots \circ \mech_n (D_n)$ satisfies $(\polMinRelax, \max(\epsilon_i))$-eOSDP.
\end{theorem}
\begin{proof}
Since under eOSDP, for $D' \in \eospN{D}$, we have that $D$ has either one more or one less record than $D'$. Thus, exactly one partition changes in $D'$, i.e., $\exists j \in [n]$, s.t., $D_j \neq D'_j$ and $\forall i \in [n] \backslash \{j\}$: $D_i = D'_i$. Lastly, since any $(P_i, \epsilon_i)$-eOSDP mechanism also satisfies $(\polMinRelax, \epsilon_i)$-eOSDP and only one partition changed between $D$ and $D'$, it follows that $\mech_1 (D_i) \circ \ldots \circ \mech_n (D_n)$ satisfies $(\polMinRelax, \max(\epsilon_i))$-eOSDP.
\end{proof}

%%%%%%%%%%%%%%%%%%%%
% PROOFS
%%%%%%%%%%%%%%%%%%%%
\subsection{Proofs}
\vspace{3mm} \noindent {\bf Proof of Theorem 3.4}

\begin{theorem*}
The $\algoname{Suppress}$ algorithm with threshold $\tau$ satisfies freedom from exclusion attack for $\phi = \tau$.
\end{theorem*}

\begin{proof}
\label{proof:pdp}
Our goal is to bound the posterior odds ratio:
\small % Do this here because hte equation overflows the line
\begin{equation*}
\frac{\PrTheta(r = x | \mech(D) \in \out)}{\PrTheta(r = y | \mech(D) \in \out)} 
	 \le  \frac{\PrTheta(r = x)}{\PrTheta(r = y)} \max_{D_1, D_2} \frac{\Pr_\mech(\mech(D_1) \in \out)}{\Pr_\mech(\mech(D_2) \in \out)}
\end{equation*}
\normalsize

Where $D$ is drawn from $\theta$ and $D_1 = D  \cup \{x \}$ and $D_2 = D \cup \{y\}$, with $x$ a sensitive record and $y$ any record, i.e., $P(x) = 0$ and $P(y) \in \{0, 1\}$. Note that $ \Pr_\mech(\mech(D_1) \in \out) = \Pr_\mech(\mech(D) \in \out)$ for $\algoname{Suppress}$, since $x$ does not participate in the computations of $\mech$. If $P(y) = 0$, the output of $\algoname{Suppress}$ will remain unchanged, since the sensitive record is not considered in the computations at all, i.e.,  $\max_{D_1, D_2} \frac{\Pr_\mech(\mech(D_1) \in \out)}{\Pr_\mech(\mech(D_2) \in \out)} = 1$. If $P(y) = 1$ the record $y$ is used from the mechanism and hence by definition of PDP $\max_{D_1, D_2} \frac{\Pr_\mech(\mech(D_1) \in \out)}{\Pr_\mech(\mech(D_2) \in \out)} \leq \epsilon^\tau$, since $ \Pr_\mech(\mech(D_1) \in \out) = \Pr_\mech(\mech(D) \in \out)$ for $\algoname{Suppress}$. Thus,
\begin{equation*}
	\frac{\PrTheta(r = x | \mech(D) \in \out)}{\PrTheta(r = y | \mech(D) \in \out)}
	\le \frac{\PrTheta(r = x)}{\PrTheta(r = y)} \cdot e^\tau
\end{equation*}

\end{proof}

\end{document}